\documentclass[10pt]{IEEEtran}
\IEEEoverridecommandlockouts
\usepackage[utf8]{inputenc}
\usepackage{amsmath}
\usepackage{amsthm}
\usepackage{amssymb}
\usepackage{graphicx}
\usepackage{psfrag}
\usepackage{cite}
\usepackage{url}
\usepackage{enumerate}

 \theoremstyle{plain}
 \newtheorem{thm}{Theorem}
 \newtheorem{lem}{Lemma}
 
 \newtheorem{cor}[thm]{Corollary}
 \theoremstyle{definition}
 \newtheorem{defn}{Definition}
 \newtheorem{rem}{Remark}

 \newtheorem{assm}{Assumption}

\newcommand{\bb}[0]{\begin{bmatrix}}
\newcommand{\eb}[0]{\end{bmatrix}}
\newcommand{\be}[0]{\begin{equation}}
\newcommand{\ee}[0]{\end{equation}}
\newcommand{\ben}[0]{\begin{equation*}}
\newcommand{\een}[0]{\end{equation*}}

\newcommand{\enVert}[1]{\left\lVert#1\right\rVert}
\newcommand{\norms}[1]{\lVert#1\rVert}

\let\norm=\normB

\newcommand{\abs}[1]{\left\lvert#1\right\rvert}

\renewcommand{\Re}[0]{\mathbb R}
\renewcommand{\exp}[1]{{\bf e}^{#1}}

\newcommand{\norml}[2]{\norm{#1}_{L_{#2}}}

\title{\vspace{18pt} \bf \Large On Adaptive Control with Closed-loop Reference Models:\\ Transients, Oscillations, and Peaking }

\author{Travis~E.~Gibson, Anuradha M. Annaswamy and Eugene Lavretsky
\thanks{T.~E. Gibson and A. M. Annaswamy are with the Department
of Mechanical Engineering, Massaschusetts Institute of Technology, Cambridge,
MA, 02139 e-mail: ({tgibson@mit.edu}).}
\thanks{E.~Lavretsky is with the Boeing Company, Huntington Beach CA, 92648.}}

\begin{document}

\maketitle

\begin{abstract}
One of the main features of adaptive systems is an oscillatory convergence that exacerbates with the speed of adaptation.  Recently it has been shown that Closed-loop Reference Models (CRMs) can result in improved transient performance over their open-loop counterparts in model reference adaptive control. In this paper, we quantify both the transient performance in the classical adaptive systems and their improvement with CRMs. In addition to deriving bounds on L-2 norms of the derivatives of the adaptive parameters which are shown to be smaller, an optimal design of CRMs is proposed which minimizes an underlying peaking phenomenon. The analytical tools proposed are shown to be applicable for a range of adaptive control problems including direct control and  composite control with observer feedback. The presence of CRMs in adaptive backstepping and adaptive robot control are also discussed. Simulation results are presented throughout the paper to support the theoretical derivations.
\end{abstract}

\section{Introduction}

A universal observation in all adaptive control systems is a convergent, yet oscillatory behavior in the underlying errors.  These oscillations increase with adaptation gain, and as such, lead to constraints on the speed of adaptation. The main obvious challenge in quantification of transients in adaptive systems stems from their nonlinear nature. A second obstacle is the fact that most adaptive systems possess an inherent trade-off between the speed of convergence of the tracking error and the size of parametric uncertainty. In this paper, we overcome these long standing obstacles by proposing an adaptive control design that judiciously makes use of an underlying linear time-varying system, analytical tools that quantify oscillatory behavior in adaptive systems, and the use of tools for decoupling speed of adaptation from parametric uncertainty.

The basic premise of any adaptive control system is to have the output of a plant follow a prescribed reference model through the online adjustment of control parameters. Historically, the reference models in Model Reference Adaptive Control (MRAC) have been open-loop in nature (see for example, \cite{annbook,ioabook}), with the reference trajectory generated by a linear dynamic model, and unaffected by the plant output. The notion of feeding back the model following error into the reference model was first reported in \cite{lee97}  and more recently in \cite{eug10aiaa,lav12tac,ste10,ste11,gib12,gib13acc1,gib13acc2,gib13ecc}. Denoting the adaptive systems with an Open-loop Reference Model as ORM-adaptive systems and those with closed-loop reference models as CRM-adaptive systems, the design that we propose in this paper to alleviate transients in adaptive control is CRM-based adaptation. 

Following stability of adaptive control systems in the 80s and their robustness in the 90s, several attempts have been made to quantify transient performance (see for example,\cite{krs93,dat94,zan94}).  The performance metric of interest in these papers stems from either supremum or L-2 norms of key errors within the adaptive system.  In \cite{krs93} supremum and L-2 norms are derived for the model following error, the filtered model following error and the zero dynamics. In \cite{dat94} L-2 norms are derived for the the model following error in the context of output feedback adaptive systems in the presence of disturbances and un-modeled dynamics. The authors of \cite{zan94} analyze the interconnection structure of adaptive systems and discuss scenarios under which  key signals can behave poorly.

%


In addition to references \cite{krs93,dat94,zan94}, transient performance in adaptive systems has been addressed in the context of CRM adaptive systems in \cite{eug10aiaa,lav12tac,ste10,ste11,gib12,gib13acc1,gib13acc2,gib13ecc}. The results in \cite{eug10aiaa,lav12tac} focused on the tracking error, with emphasis mainly on the initial interval where the CRM-adaptive system exhibits fast time-scales. In \cite{ste10} and \cite{ste11}, transient performance is quantified using a damping ratio and natural frequency type of analysis. However, assumptions are made that the initial state error is zero and that the closed-loop system state is independent of the feedback gain in the reference model, both of which may not hold in general. 


In this paper, we start with CRM adaptive systems as the design candidate, and quantify the underlying transient performance. This is accomplished by deriving L-2 bounds on key signals and their derivatives in the adaptive system. These bounds are then related to the corresponding frequency content using a Fourier analysis, thereby leading to an analytical basis for the observed reduction in oscillations with the use of CRM. It is also shown that in general, a peaking phenomenon can occur with CRM-adaptive systems, which then is shown to be minimized through an appropriate design of the CRM-parameters. Extensive simulation results are provided, illustrating the conspicuous absence of oscillations in CRM-adaptive systems in contrast to their dominant presence in ORM-adaptive systems.  The results of this paper build on preliminary versions in \cite{gib12,gib13acc1,gib13acc2} where the bounds obtained were conservative. While all results derived in this paper are applicable to plants whose states are accessible for measurement, we refer the reader to  \cite{gib13ecc} for extensions to output feedback. 

This paper also addresses Combined/composite direct and indirect Model Reference Adaptive Control (CMRAC)  \cite{duarte1989combined,slotine1989composite}, which is another class of adaptive systems in which a noticeable improvement in transient performance was demonstrated. While the results of
these papers established stability of combined schemes, no
rigorous guarantees of improved transient performance were
provided, and have remained a conjecture \cite{eugeneTAC09}. We
introduce CRMs into the CMRAC and show how improved transients can be guaranteed. We close this paper with a discussion of CRM and related concepts that appear in other adaptive systems as well, including nonlinear adaptive control
\cite{kkkbook} and adaptive control in robotics \cite{slobook}.

This paper is organized as follows. Section II contains the
basic CRM structure with L-2 norms of the key signals in the
system. Section III investigates the peaking in the reference
model. Section IV contains the multidimensional states accessible extension. Section V investigates composite control structures with CRM. Section VI
explores other forms of adaptive control where closed
loop structures appear.

\section{CRM-Based Adaptive Control of Scalar Plants}\label{sec:crm}
Let us begin with a scalar system,
\be\label{eqd:plant}
\dot x_p(t) = a_p x_p(t) + k_p u(t) 
\ee
where ${x_p(t)\in \Re}$ is the plant state, ${u(t)\in \Re}$ is the control input, ${a_p\in\Re}$ is an unknown scalar and only the sign of ${k_p\in\Re}$ is known. 
We choose a closed-loop reference model as 
\be\label{eqd:reference}
\dot x_m(t) =  a_m x_m(t) + k_m r(t) -\ell (x(t)-x_m(t)) . 
\ee
All of the parameters above are known and scalar, $x_m(t)$ is the reference model state, $r(t)$ is a bounded reference input and $a_m,\ell<0$ so that the reference model and the subsequent error dynamics are stable. The open-loop reference model dynamics
\be\label{orm}
 \dot x_m^o(t) = a_m x_m^o(t) + k_m r(t)
\ee
is the corresponding true reference model that we actually want $x_p$ to converge to.

The control law is chosen as
\be\label{eqd:controller} u(t) = \bar\theta^T(t) \phi(t) \ee
where we have defined
${ \bar\theta^T(t) = \bb \theta(t) & k(t)\eb^T}$ and  ${\phi^T(t)=\bb x_p(t) & r(t)\eb^T}$
with an update law
\be\label{eqd:update}
\dot{\bar \theta}=  -\gamma \text{sgn}(k_p)e \phi 
\ee
where $\gamma>0$ is a free design parameter commonly referred to as the adaptive tuning gain and $e(t)=x_p(t)-x_m(t)$ is the state tracking error. From this point forward we will suppress the explicit time dependance of parameters accept for emphasis.

We  define the  parameter error
${\tilde{\bar\theta}(t)= \bar\theta(t)-\bar\theta^*}$, where ${\bar\theta^*\in\Re^2}$ satisfies  ${\bar\theta^{*T} = \bb \frac{a_m-a_p}{k_p} & \frac{k_m}{k_p}\eb^T}$.
The corresponding closed loop error dynamics are:
\be\label{eqd:error}
\dot e(t) =  (a_m+\ell) e + k_p \tilde {\bar\theta}^T \phi.
\ee

\subsection{Stability Properties of CRM-adaptive systems}
Theorem \ref{thm1} establishes the stability of the above adaptive system with the CRM:
\begin{thm}\label{thm1}
The adaptive system with the plant in \eqref{eqd:plant}, with the controller defined by \eqref{eqd:controller}, the update law in \eqref{eqd:update} with the reference model as in \eqref{eqd:reference} is globally stable, $\lim_{t\to\infty} e(t) = 0$, and
\be\label{l2e}
\norml{e}{2}^2 \leq  \frac{1}{\abs{a_m+\ell}} \left(\frac{1}{2} e(0)^2 + \frac{\abs{k_p} }{2\gamma}\tilde{\bar\theta}^T(0) \tilde{\bar\theta}(0) \right).
\ee
 \end{thm}
\begin{proof}
Consider the lyapunov candidate function
\ben
V(e(t),\tilde\theta(t)) = \frac{1}{2} e^2 + \frac{\abs{k_p} }{2\gamma} \tilde{\bar\theta}^T\tilde{\bar\theta}. 
\een
Taking the time derivative of $V$ along the system directions we have 
$\dot V = (a_m+\ell) e^2 \leq 0$. Given that $V$ is positive definite and $\dot V$ is negative semi-definite we have that ${V(e(t),\tilde{\bar\theta}(t)\leq V(e(0),\tilde{\bar \theta}(0))<\infty}$.
Thus $V$ is bounded and this means in tern that $e$ and $\tilde{\bar\theta}$ are bounded, with 
\be\label{elinf}
\norml{e(t)}\infty^2 \leq 2 V(0).
\ee

Given that $r$ and $e$ are bounded and the fact that $a_m<0$, the reference model is stable. Thus we can conclude $x_m$, and therefore $x_p$, are bounded. Given that $\bar\theta^*$ is a constant we can conclude that $\bar\theta$ is bounded from the boundedness of $\tilde{\bar\theta}$. This can be compactly stated as $e, x_p, \tilde{\bar\theta}, \bar\theta \in \mathcal L_\infty$, and therefore all of the signals in the system are bounded.

In order to prove asymptotic stability in the error we begin by noting that $
- \int_0^t \dot V = V(e(0),\tilde\theta(0)) - V(e(t),\tilde\theta(t)) \leq V(e(0),\tilde\theta(0)) $. This in turn can be simplified as ${\abs{a_m+\ell} \int_0^t e(t)^2 \leq  V(0)}\  \forall \  {t\geq 0}
$. Dividing by $\abs{a_m+\ell}$ and taking the limit as $t\to\infty$ we have
\be\label{el2}
\norml{e}{2}^2 \leq  \frac{V(0)}{\abs{a_m+\ell}}
\ee
which implies \eqref{l2e}. Given that $e\in \mathcal L_2 \cap \mathcal L_\infty$ and $\dot e \in \mathcal L_\infty$, Barbalat's Lemma is satisfied and therefore $\lim_{t \to \infty}e(t)=0$ \cite{annbook}.
\end{proof}

Theorem \ref{thm1} clearly shows that CRM ensures stability of the adaptive system. Also, from the fact that $e\in\mathcal L_2$ we have that $x_m(t) \to x_m^o(t)$ as $t\to\infty$  \cite[\S IV.1, Theorem 9(c)]{desbook}. That is, the closed-loop reference model asymptotically converges to the open-loop reference model, which is our true desired trajectory. The questions that arise then is one, if any improvement is possible in the transient response with the inclusion of $\ell$, and second, how close is $x_m(t)$, the response of the CRM in relation to that of the original reference model, $x_m^o(t)$. These are addressed in the following section.

\subsection{Transient Performance of CRM-adaptive systems}

The main impact of the CRM is a modification in the realization of the reference trajectory, from the use of a linear model to a nonlinear model. This in turn produces a more benign target for the adaptive closed-loop system to follow, resulting in better transients. It could be argued that the reference model meets the closed-loop system half-way, and therefore reduces the burden of tracking on the adaptive system and shifts it partially to the reference model. In what follows, we precisely quantify this effect.

As Equation \eqref{l2e} in Theorem \ref{thm1} illustrates, the L-2 norm of $e$ has two components, one associated with the initial error in the reference model, $e(0)$, and the other with the initial error in the parameter space, $\tilde{\bar\theta}(0)$. The component associated with $\tilde{\bar\theta}(0)$ is inversely proportional to the product ${\gamma \abs{\ell}}$ and the component associated with the initial model following error $e(0)$ is inversely proportional to $\abs{\ell}$ alone. Therefore, without the use of the feedback gain $\ell$ it is not possible to uniformly decrease the L-2 norm of the model following error. This clearly illustrates the advantage of using the CRM over the ORM, as in the latter, $\ell=0$.

While CRM-adaptive systems bring in this obvious advantage, they can also introduce an undesirable peaking phenomenon. In what follows, we introduce a definition and show how through a proper choice of the gain $\ell$, this phenomenon can be contained, and lead to better bounds on the parameter derivatives.
As mentioned in the introduction, we quantify transient performance in this paper by deriving L-2 bounds on the parameter derivative $\dot{\bar\theta}$, which in turn will correlate to bounds on the amplitude of frequency oscillations in the adaptive parameters. For this purpose, we first discuss the L-2 bound on $e$ and supremum bound for $x_m$. We then describe a peaking phenomenon that is possible with CRM-adaptive systems.



\subsubsection{L-$\infty$ norm of $x_m$}
The solution to the ODE in \eqref{eqd:reference} is 
\be\label{xm1}\begin{split}
x_m(t) =& \exp{a_m t} x_m(0) + \int_0^t \exp{a_m(t-\tau)} r(\tau) d\tau  \\  & - \ell \int_0^t \exp{a_m(t-\tau)} e(\tau) d\tau.
\end{split}\ee
The solution to the ODE in \eqref{orm} is 
\be
x_m^o(t)=\exp{a_m t} x_m^o(0) + \int_0^t \exp{a_m(t-\tau)} r(\tau) d\tau.
\ee
For ease of exposition and comparison, $x_m(0) = x_m^o(0)$ and thus
\be\label{xmxmo}
{x_m(t)} =  {x_m^o(t)} -\ell \int_0^t \exp{a_m(t-\tau)}{e(\tau)} d\tau.
\ee
Denoting the difference between the open-loop and closed-loop reference model as $
{\Delta x_m = x_m-x_m^o}$, using  Cauchy Schwartz Ineqality on ${\int_0^t \exp{a_m(t-\tau)} \norm{e(\tau)} d\tau}
$, and the bound for $\norml{e(t)} 2$ from \eqref{el2}, we can conclude that
\be\label{xm3}
\norm{\Delta x_m(t)} \leq \abs{\ell} \sqrt{\frac{1}{\abs{2 a_m}}} \sqrt{\frac{V(0)}{\abs{a_m+\ell}}}.
\ee

We quantify the peaking phenomenon through the following definition:
\begin{defn}
Let $\alpha\in \Re^+$, $a_1,a_2$ are fixed positive constants, $x: \Re^+\times\Re^+ \to\Re$ and $x(\alpha;t)\in \mathcal L_2$,
\ben
y(\alpha;t) \triangleq \alpha \int_0^t \exp{-a_1(t-\tau)} x(\alpha;\tau) d\tau,
\een
Then the signal $y(\alpha;t)$ is said to have a peaking exponent $s$ 
with respect to $\alpha$ if\ben
\norml{y(t)}\infty \leq a_1 \alpha^s+a_2.
\een
\end{defn}

\begin{rem}We note that this definition of peaking differs from that of peaking for linear systems given in \cite{sus91}, and references there in. In these works, the underlying peaking behavior corresponds to terms of the form $\kappa \exp{-\alpha t}$, $\alpha,\kappa > 0$, where any increase in $\alpha$ is accompanied by a corresponding increase in $\kappa$ leading to peaking. This can occur in linear systems where the Jacobian is defective \cite{mol03}. In contrast, the peaking of interest in this paper occurs in adaptive systems where efforts to decrease the $L_2$ norm of $x$ through the increase of $\alpha$ leads to the increase of $y$ causing it to peak. This is discussed in detail below.
\end{rem}

From Eq. (12), it follows that $\Delta x_m$ can be equated with $y$ and $e$ with $x$ in 
Definition 1. Expanding $V(0)$, the bound on $\Delta x_m(t)$ in  (13) can be represented as
\ben\norm {\Delta x_m(t)} \leq b_1 |\ell|^{1/2} + b_2\left(\frac{\abs \ell}{\gamma}\right)^{1/2} \een
where $
b_1= \sqrt{\frac{e(0)^2}{2|a_m|}}$ and $b_2= \sqrt{\frac{\norm{\tilde{\bar{\theta}}(0)}^2}{2|a_m|}}$. We note that $\gamma$ is a free design parameter in the adaptive system. Therefore, one can choose $\gamma=\abs{\ell}$ and achieve the bound
\be\label{tr1}\norm {\Delta x_m(t)} \leq b_1 |\ell|^{1/2} + b_2.\ee
From \eqref{tr1} and Definition 1, it follows that with $\gamma=O(| \ell |)$, 
 $\Delta x_m$ has a peaking exponent of 0.5 with respect to $|\ell |$. Similar to \eqref{tr1} the following bound holds for $x_m$:
\be\label{tr2}
\norml{ x_m(t)}\infty \leq b_1 |\ell|^{1/2} + b_3
\ee
where $b_3 = b_2+\norml{x_m^o(t)}\infty$ and $\gamma=\abs \ell$. That is the bounds in \eqref{tr1} and \eqref{tr2} increase with $| \ell |$, which implies that $\Delta x_m(t)$ and 
therefore $x_m(t)$ can exhibit peaking. 

While it is tempting to
simply pick ${e(0) = 0}$ so that $b_1=0$, as is suggested in [6], [7] to circumvent
this problem, it is not always possible to do so, as $x(0)$ may not be available as a measurement because of noise or disturbance that may be present. In Section III, we present an approach where tighter bounds for $x_m(t)$ are derived, which enables us to reduce the peaking exponent of $\Delta x_m$ from 0.5 to zero.


Before moving to the L-2 bounds on $\dot k$ and $\dot \theta$, we motivate the importance of L-2 bounds on signal derivatives and how they relate to the frequency characteristic of the signals of interest. We use a standard property of Fourier series and continuous functions \cite{lax book,rud76} summarized in Lemma \ref{lemmer} and Theorem \ref{thm:p} below:
\begin{lem}\label{lemmer}
Consider a periodic signal ${f(t)\in \Re}$ over a finite interval ${\mathbb T = [t_1, t_1+\tau]}$ where $\tau$ is the period of $f(t)$. The Fourier coefficients of $f(t)$ are then given by ${F(n)=\frac{1}\tau\int_{\mathbb T} f(t) \exp{-i \omega(n)t} dt}$ with $\omega(n)=2\pi n/\tau$. 
If $f(t) \in \mathcal C^1$, given $\epsilon_1>0$, there exists an integer $N_1$ such that 
\begin{itemize}
\item[(i)]
$
\quad \enVert{f(t)-\sum_{n = -{N_1}}^{N_1} F(n) \exp{i \omega(n) t}} \leq \epsilon_1,\qquad \forall t\in \mathbb T.
$ \end{itemize}
If in addition $f(t)\in \mathcal C^2$, then for all $\epsilon_2>0$ there exists an integer $N_2$ such that
\begin{itemize}
\item[(ii)]$
\enVert{\frac{d}{dt} f(t)- \sum_{n = -{N_2}}^{N_2}  i\omega(n) F(n) \exp{i \omega(n) t}} \leq \epsilon_2, \forall t\in \mathbb T.
$\end{itemize}
\end{lem}

\begin{rem} We note that one can use the notion of generalized functions as presented in \cite{lig58} in order to obtain Fourier approximations with assumptions of the interval being finite and periodicity relaxed.
\end{rem}

\begin{thm} \label{thm:p}If $f(t)\in\mathcal C^2$ and periodic with period $\tau$, then the following equality holds
\be\label{dd1}
\int_{\mathbb T} \norm{\dot f(t)}^2 dt= \sum_{n=-\infty}^{\infty}  \abs{F(n)}^2  \abs{\omega(n) 2 \pi n}
\ee
where ${F(n)=\frac 1 \tau \int_{\mathbb T} f(t) \exp{-i \omega(n)t} dt}$, $\omega(n)=2\pi n/\tau$ and $\mathbb T = [t_1, t_1+\tau]$.
\end{thm}
\begin{proof}
This follows from Parseval's Theorem. From Lemma \ref{lemmer}(ii), 
$\int_{\mathbb T} \norm{\dot f(t)}^2dt = \int_{\mathbb T} \enVert{ \sum_{n=-\infty}^{\infty}i\omega(n) F(n) \exp{i \omega(n) t}}^2 dt.$
Using the orthogonality of $\exp{i\omega(n)t}$ we have that ${\int_{t_1}^{t_1+\tau} \exp{i \omega(n) t}\exp{-i \omega(m) t} dt=0}$ for all integers ${m\neq n}$. It also trivially holds that $\int_{t_1}^{t_1+\tau} \exp{i \omega(n) t}\exp{-i \omega(n) t}dt = \tau$. Using these two facts along with the fact that the convergence in Lemma \ref{lemmer}(ii) is uniform, the integral above can be simplified as
$\label{dd3}\int_{\mathbb T} \norm{\dot f(t)}^2dt =   \sum_{n=-\infty}^{\infty} \omega(n)^2 \abs{F(n)}^2 \tau$.
Expanding one of the $\omega(n)$ terms and canceling the $\tau$ term gives us \eqref{dd1}.

\begin{rem}
From Theorem \ref{thm:p} it follows that when the L-2 norm of the derivative of a function is reduced, the product $\abs{F(n)}^2  \abs{\omega(n) 2 \pi n}$ is reduced for all $n\in\mathbb Z$. Given that $\omega(n)$ is the natural frequency for each Fourier approximation and $\abs{F(n)}$ their respective amplitudes, reducing the L-2 norm of the derivative of a function implicitly reduces the the amplitude of the high frequency oscillations.
\end{rem}

\end{proof}

\subsubsection{L-2 norm of $\dot k,\dot \theta$}
With the bounds on $e$ and $x_m$ in the previous sections, we now derive bounds on the adaptive parameter derivatives.
From \eqref{eqd:update} we can deduce that $\norm{\dot k}^2  = \gamma^2 e^2 r^2$.
Integrating both sides and taking the supremum of $r$ we have
\be\label{reffer1}
\int_0^t\norm{\dot k(\tau)}^2 d\tau   \leq \gamma^2   \norml{r(t)}{\infty}^2 \norml{e(t)}{2}^2. 
\ee
Using the bound on  $\norml{e}2$ from \eqref{el2} we have that
\be\label{kdotl2}
{\norml{\dot k(t)}{2}^2 \leq  \frac{2 \gamma^2  \norml{r(t)}{\infty}^2 V(0)}{\abs{a_m+\ell}}}.
\ee
Similarly, from \eqref{eqd:update} we can derive the inequality
%
\be\begin{split}\label{above2}
\int_0^t\norm{\dot\theta(\tau)}^2 d\tau \leq & 2  \gamma^2 \norml{e(t)}\infty ^2 \int_0^t e(\tau)^2  d\tau \\ & +2 \gamma^2 \norml{x_m(t)}\infty^2 \int_0^te(\tau)^2  d\tau.
\end{split}\ee
Using the bounds for $\norml{e(t)}\infty$ in \eqref{elinf}, $\norml{e(t)}2$ in  \eqref{el2}, and the following bound on 
\be\norml{x_m(t)}\infty^2 \leq 2\norml{x^o_m(t)}\infty^2  + \frac{\abs{\ell}^2V(0)}{\abs{a_m}\abs{a_m+\ell}},\ee which follows from the bound on $\Delta x_m(t)$ in \eqref{xm3}, the bound in \eqref{above2} can be simplified as
\be \label{thetadotl2}
\begin{split}\norml{\dot\theta(t)} 2^2 \leq &  4\gamma^2 \frac{V(0) \norml{x_m^o(t)}\infty^2}{\abs{a_m+\ell}} + 4 \gamma^2 \frac{V(0)^2}{\abs{a_m+\ell}} \\ &+ 2\gamma^2 \frac{\abs \ell ^2}{\abs{a_m}} \frac{V(0)^2}{\abs{a_m+\ell}^2}.
\end{split}
\ee

From \eqref{kdotl2} it is clear that by increasing $\abs{\ell}$ one can arbitrarily decrease the L-2 norm of $\dot k$. The same is not true, however, for the L-2 norm of $\dot\theta$ given in \eqref{thetadotl2}. Focusing on the first two terms we see that their magnitude is proportional to $\gamma^2/\abs{a_m+\ell}$. Letting $\ell$ approach negative infinity, the first and second second terms in $\eqref{thetadotl2}$ approaches zero and the third term converges to a bound which is proportional to $\gamma^2V(0)^2$. When $\ell=0$, the second term becomes proportional to $\gamma^2V(0)^2$ and the last term in \eqref{thetadotl2} becomes zero. From the previous discussion it is clear that regardless of our choice of $\ell$, the only way to uniformly decrease the the L-2 norms of the derivatives of the adaptive terms is by decreasing $\gamma$. This leads to the classic trade-off present in adaptive control. One can reduce the high frequency oscillations in the adaptive parameters by choosing a small $\gamma$, this however leads to poor reference model tracking.  This can be scene by expanding the bound on $\norml{e(t)}\infty$ in \eqref{elinf},
\be\label{elinf2}
\norml{e(t)}\infty^2 \leq e(0)^2+ \frac{1}{\gamma}\tilde{\bar\theta}(0)^T\tilde{\bar\theta}(0).
\ee
If one chooses a small $\gamma$, then poor state tracking performance can occur, as the second term in \eqref{elinf2} is large for small $\gamma$. Therefore it still remains to be seen as to how and when CRM leads to an advantage over ORM. As shown in the following subsection and subsequent section, this can be demonstrated through the introduction of projection in the adaptive law and a suitable choice of $\ell$ and $\gamma$. This in turn will allow the reduction of high frequency oscillations.

\subsection{Effect of Projection Algorithm}
It is well known that some sort of modification of the adaptive law is needed to ensure boundedness in the presence of perturbations such as disturbances or unmodeled dynamics. We use a projection algorithm \cite{pom92} with CRMs as
\be\label{ad:proj}
\dot{\bar\theta}= \text{Proj}_\Omega(-\gamma \text{sgn}(k_p) e \phi,\bar\theta)
\ee
where $\bar\theta(0),\bar\theta^* \in \Omega$, with $\Omega\in\Re^2$ a closed and convex set  centered at the origin whose size is dependent on a known bound of the parameter uncertainty $\bar\theta^*$. Equation \eqref{ad:proj} assures that ${\bar\theta(t)\in\Omega\ \forall\ t\geq0}$ \cite{pom92}. The following definition will be used throughout:
\be
 \Theta_\text{max}\triangleq \sup_{\bar\theta,\bar\theta^*\in\Omega} \norm{\tilde{\bar\theta}} . 
\ee

Beginning with the already proven fact that $\dot V \leq (a_m+\ell) e^2$, we note that the following bound holds as well with the use of \eqref{ad:proj}:
\be\label{gb1}
\dot V(t) \leq - 2{\abs{a_m+\ell}} V + \frac{\abs{a_m+\ell}}{\gamma} \abs{k_p} \Theta_\text{max}^2.
\ee
Using Gronwall-Bellman \cite{bell43} it can be deduced that
\be
V(t) \leq \left(V(0)-\frac{\abs{k_p}} {2\gamma}\Theta_\text{max}^2\right) \exp{-2\abs{a_m+\ell}t} +  \frac{\abs{k_p}}{2\gamma} \Theta_\text{max}^2
\ee
which can be further simplified as
\be\label{exp2}
V(t) \leq \frac{1}{2}\norm{e(0)}^2 \exp{-2\abs{a_m+\ell}t} +  \frac{\abs{k_p}}{2\gamma} \Theta_\text{max}^2
\ee
which informs the following exponential bound on $e(t)$:
\be\label{ert}
\norm{e(t)}^2 \leq \norm{e(0)}^2 \exp{-2\abs{a_m+\ell}t} +  \frac{\abs{k_p}}{\gamma} \Theta_\text{max}^2.
\ee

The discussions in Section \ref{sec:crm} show that with a projection algorithm, the CRM adaptive system is not only stable but satisfies the transient bounds in \eqref{el2}, \eqref{tr2},  \eqref{kdotl2}, \eqref{thetadotl2}, \eqref{exp2} and \eqref{ert}. The bounds in \eqref{tr2},   \eqref{kdotl2} and \eqref{thetadotl2} leave much to be desired however, as it is not clear how the free design parameters $\ell$ and $\gamma$ can be chosen so that the bounds on $\norml{\dot k}2$ and $\norml{\dot \theta}2$ can be systematically reduced while simultaneously controlling the peaking in the reference model output $x_m$. Using the bounds in \eqref{exp2} and \eqref{ert}, in the following section,  we propose an ``optimal'' CRM design that does not suffer from the peaking phenomena, and show how the bounds in \eqref{tr2}, \eqref{kdotl2} and \eqref{thetadotl2} can be further improved. We also make a direct connection between the L-2 norm of the derivative of a signal, and the frequency and amplitude of oscillation in that signal.

\section{Bounded Peaking with CRM adaptive systems}
\subsection{Bounds on $x_m$}
We first show that the peaking that $\norm{x_m(t)}$ was shown to exhibit in Section II-A can be reduced through the use of a projection algorithm in the update law as in \eqref{ad:proj}, and a suitable choice of $\gamma$ and $\ell$. For this purpose we derive two different bounds, one over the time interval $[0,t_1]$ and another over $[t_1,\infty)$.

\begin{lem}\label{lem1}
Consider the adaptive system with the plant in \eqref{eqd:plant}, with the controller defined by \eqref{eqd:controller}, the update law in \eqref{ad:proj} with the reference model as in \eqref{eqd:reference}. For all $\delta>1$ and $\epsilon>0$, there exists a time $t_1\geq0$ such that
\be\begin{split}\label{xm11}
\norm{x_p(t)} &\leq \delta \norm{x_p(0)} + \epsilon \norml{r(t)}\infty   \\
\norm{x_m(t)}& \leq \delta \norm{x_p(0)} + \epsilon \norml{r(t)}\infty+ \sqrt{2V(0)} \\ 
\end{split}\ee
$\forall\   0\leq t\leq t_1$.
\end{lem}
\begin{proof}
The plant in \eqref{eqd:plant} is described by the dynamical equation
\ben
\dot x_p = (a_m +k_p\tilde\theta) x_p + k_p \tilde k r
\een
where we note that $(a_m +k_p\tilde\theta) $ can be positive. This leads to the inequality
\ben\begin{split}
\norm{x_p(t)} \leq &\norm{x_p(0)} \exp{(a_m +\abs{k_p}\Theta_\text{max})t}  \\ &+ \int_0^t\exp{(a_m +\abs{k_p}\Theta_\text{max})(t-\tau) }\abs{k_p}\Theta_\text{max}\norml{r(\tau)}\infty d\tau.
\end{split}\een
For any ${\delta>}1$ and any ${\epsilon>0}$, it follows from the above inequality that a $t_1$ exists such that 
$\exp{(a_m +\abs{k_p}\Theta_\text{max})t}\leq \delta$ and ${\int_0^t\exp{(a_m +\abs{k_p}\Theta_\text{max})(t-\tau) }\abs{k_p}\Theta_\text{max} d\tau}\leq \epsilon$, ${\forall\ t\leq t_1}$. The bound on $x_m(t)$ follows from the fact that $\norm{x_m}\leq \norm{x_p}+\norm e$ and from \eqref{elinf}.
\end{proof}

\begin{rem}
The above lemma illustrates the fact that if $t_1$ is small, the plant and reference model states cannot move arbitrarily far from their respective initial conditions over $[0,\ t_1]$.
\end{rem}

\begin{lem}\label{lem:exp}
For any $a\geq0$ $\exists$ an $x^*<0$ such that for all ${x\leq x^*<0}$
\ben
\exp{xa}\leq \abs x^{-\frac 1 2} \quad \forall\ x\leq x^*<0.
\een
\end{lem}
\begin{proof}
Exponential functions with negative exponent decay faster than any fractional polynomial.
\end{proof}

We now derive bounds on $x_m(t)$ when $t\geq t_1$. For this purpose a tighter bound on the error $e$ than that in \eqref{el2} is first derived.

\begin{lem}\label{et1}
Consider the adaptive system with the plant in \eqref{eqd:plant}, with the controller defined by \eqref{eqd:controller}, the update law in \eqref{ad:proj} with the reference model as in \eqref{eqd:reference}. Given a time $t_1\geq0$, there exists an $\ell^*$ s.t.
\be\label{et22}
\sqrt{\int_{t_1}^\infty \norm{e}^2 d\tau} \leq  \frac{\norm{e(0)} }{\sqrt{2}\abs{a_m+\ell}} + \sqrt{\frac{\abs{k_p}}{2\gamma \abs{a_m+\ell}}} \Theta_\text{max}
\ee
for all $\ell\leq \ell^*$.
\end{lem}
\begin{proof}
Substitution of ${t=t_1}$ in \eqref{exp2} 
and using the fact that ${\dot V(t)= -{\abs{a_m+\ell}}\norm{e(t)}^2}$, the following bound is obtained:
\be\begin{split}
\int_{t_1}^\infty \norm{e}^2 d\tau \leq & \frac{\norm{e(0)}^2 \exp{-2\abs{a_m+\ell}t_1}}{2\abs{a_m+\ell}} \\&+  \frac{\abs{k_p}}{2\gamma \abs{a_m+\ell}} \Theta_\text{max}^2.
\end{split}\ee
Noting that $\sqrt{\exp{-2 \abs{a_m+\ell}t_1}} =\exp{-\abs{a_m+\ell}t_1}$, and using the result from Lemma \ref{lem:exp}, we know that there exists an $\ell^*$ such that for all $\ell<\ell^*$, $\exp{(a_m+\ell)t_1}\leq \abs{a_m+\ell}^{-1/2}$. This leads to \eqref{et22}.
\end{proof}

Similar to the definition of $\Delta x_m(t)$  in Section II.B we define \ben
\Delta \bar x_m(t)\triangleq \abs{\ell} \int_{t_1}^t \exp{-{\abs{a_m}}(t-\tau)} \norm{e(\tau)} d\tau\een for all $t\geq t_1$.
Choosing $\ell\leq \ell^*$ with $\ell^*$ defined in Lemma \ref{et1}, using the bound on $e(t)$ in \eqref{et22} and the Cauchy Schwartz inequality, we have that
\ben
\norml{\Delta\bar x_m(t)}\infty \leq b_4 + b_5 \left({\frac{\abs \ell}{\gamma}}\right)^{1/2} \quad \forall t\geq t_1 
\een
where $b_4=\frac{\norm{e(0)}}{2 \sqrt{\abs{a_m}}}$ and $b_5 = \frac{\sqrt{\abs{k_p}}\Theta_\text{max}}{2\sqrt{\abs{a_m}}}$. Choosing $\gamma = \abs \ell$ the bound above becomes
\be\label{t51}
\norml{\Delta\bar x_m(t)}\infty \leq b_4 + b_5 \quad t\geq t_1.
\ee
Comparing the bound in \eqref{t51} to the bound in \eqref{tr1}, we note that the peaking exponent (Definition 1) has been reduced from 1/2 to 0 for the upper bound on the convolution integral of interest. Thus, as $\abs{\ell}$ is increased, the term $\Delta \bar x_m(t)$ will not exhibit peaking. That is, the response of the CRM is fairly close to that of the ORM. This result allows us to obtain a bound on the closed-loop reference model $x_m(t)$ that does not increase with increasing $\abs{\ell}$. This is explored in the following theorem and subsequent remark in detail.

\begin{thm}\label{thm:tt}Consider the adaptive system with the plant in \eqref{eqd:plant}, the controller defined by \eqref{eqd:controller}, the update law in \eqref{ad:proj} with the reference model as in \eqref{eqd:reference}, with $t_1$ chosen as in Lemma \ref{lem1} and $\ell\leq\ell^*$ where $\ell^*$ satisfies \eqref{et22}. It can then be shown that
\be\label{xmtight}
\norm{x_m(t)}_{t \geq t_1}^2 \leq c_1(t) + \frac{\norm{e(0)}^2}{\abs{a_m}}+\frac{\abs{\ell} \abs{k_p}\Theta_\text{max}^2}{\gamma\abs{a_m}}
\ee
where
\ben
c_1  \triangleq  2\left(\norml{x_m^o}\infty + \norm{ x_m(t_1)}\exp{-\abs{a_m}(t-t_1)} \right)^2. 
\een
\end{thm}
\begin{proof}
The solution of \eqref{eqd:reference} for $t \geq t_1$ is given by
\ben\begin{split}
\norm{x_m(t)} \leq & \norml{x_m^o}\infty + \norm{ x_m(t_1)}\exp{-\abs{a_m}(t-t_1)} \\&+ \abs{\ell} \int_{t_1}^\infty \exp{-{\abs{a_m}}(t-\tau)} \norm{e(\tau)} d\tau.
\end{split}\een
Using the Cauchy Schwartz Inequality and \eqref{et22} from Lemma \ref{et1},
we have that
\ben\begin{split}
\norm{ x_m(t)}_{t\geq t_1} \leq & \norml{x_m^o}\infty+ \norm{\ x_m(t_1)}\exp{-\abs{a_m}(t-t_1)} \\&+ \frac{\abs{\ell}}{\sqrt{2\abs{a_m}}} \cdot \\ & \left(\frac{\norm{e(0)}}{\sqrt 2 \abs{a_m+\ell}}+\frac{\sqrt{\abs{k_p}}\Theta_\text{max}}{\sqrt{2\gamma\abs{a_m+\ell}}}\right),
\end{split}\een
for all $\ell<\ell^*$. Squaring leads to \eqref{xmtight}.
\end{proof}

\begin{cor}
Following the same assumptions as Theorem \ref{thm:tt}, with $\gamma=\abs{\ell}$
\begin{align}
\label{xmtighter}
\norm{x_m(t)}^2_{t\leq t_1}& \leq 2(\delta \norm{x_p(0)} + \epsilon \norml{r(t)}\infty)^2+ 4V(0) \\ 
\label{xmtighter2}
\norm{x_m(t)}_{t \geq t_1}^2& \leq c_1(t) + \frac{\norm{e(0)}^2}{\abs{a_m}}+\frac{ \abs{k_p}\Theta_\text{max}^2}{\abs{a_m}}
\end{align}
\end{cor}

\begin{rem}
Through the use of a projection algorithm in the adaptive law, the exploitation of finite time stability of the plant in Lemma 1, and through the use of the extra degree of freedom in the choice of $\ell$, we have obtained a bound for $\norm{x_m(t)}^2_{t\leq t_1}$ in \eqref{xmtighter} which is only a function of the initial condition of the plant and controller. Similarly for $\norm{x_m(t)}^2_{t\geq t_1}$, we have derived a bound in \eqref{xmtighter} which is once again a function of the initial condition of the plant and controller alone. The most important point to note is that unlike \eqref{tr2}, the bound on $x_m$ in \eqref{xmtighter} and \eqref{xmtighter2} is no longer proportional to $\ell$ in any power. This implies that even for large $\abs{\ell}$, an appropriate choice of the adaptive tuning parameter $\gamma$ can help reduce the peaking in the reference model. This improvement was possible only through the introduction of projection and the use of the Gronwall-Bellman inequality.
\end{rem}

%

\subsection{Bounds on parameter derivatives and oscillations}
We now present the main result of this paper. 

\begin{thm}\label{thm3}
The adaptive system with the plant in \eqref{eqd:plant}, the controller defined by \eqref{eqd:controller}, the update law in \eqref{ad:proj} with the reference model as in \eqref{eqd:reference}, with $t_1$ chosen as in Lemma \ref{lem1} and $\ell\leq\ell^*$ where $\ell^*$ is given in Lemma \ref{et1} and $\gamma=\abs{\ell}$, the following bounds are satisfied for all $\gamma \geq 1$:
\be\label{t1infb}
\begin{split}
\int_{t_1}^\infty \norm{\dot k}^2 d\tau\leq &  \left(\norm{e(0)}^2+\abs{k_p}\Theta_\text{max}^2\right) \norml{r(t)}\infty \\
\int_{t_1}^\infty \norm{\dot \theta}^2 d\tau \leq &  \left(\norm{e(0)}^2+\abs{k_p}\Theta_\text{max}^2\right) c_2 \\
   & +  \left(\norm{e(0)}^2+\abs{k_p}\Theta_\text{max}^2\right) \left(\frac{c_3}{\sqrt{\abs \ell}} + \frac{c_4}{\gamma}\right) \end{split} \ee 
where $c_2,c_3,c_4$ are independent of $\gamma$ and $\ell$, and are only a function of the initial conditions of the system and the fixed design parameters.
\end{thm}
\begin{proof}
Using \eqref{reffer1} and  \eqref{et22}, together with the fact that $\gamma=\abs \ell$, we obtain the first inequality in \eqref{t1infb}. To prove the bound on $\dot \theta$ we start with \eqref{above2}, and note that 
\be
\label{tt22}\gamma^2 \int_{t_1}^\infty e(\tau)^2 d\tau \leq \left(\norm{e(0)}^2+\abs{k_p}\Theta_\text{max}^2\right). 
\ee
Using the bound in \eqref{tt22} and setting $c_2= \norm{x_m(t)}^2_{t\geq t_1}$ from \eqref{xmtight} we have the first term in the bound on $\dot\theta$ in \eqref{t1infb}.

We note from \eqref{ert} and Lemma \ref{lem:exp} that
\ben
\norm{e(t)} \leq \frac{e(0)^2}{\sqrt{\abs{a_m+\ell}}}+\frac{\abs{k_p}\Theta_\text{max}}{\gamma} \qquad \forall t\geq t_1.
\een
This together with \eqref{tt22} leads to the second term in the bound on $\dot\theta$ in \eqref{t1infb}.
Therefore, $c_2,c_3$ and $c_4$ are independent of $\gamma$ and $\ell$.
\end{proof}

\begin{rem}
From the above Theorem it is clear that if $\gamma$ and $\abs{\ell}$ are increased while holding $\gamma=\abs{\ell}$, the L-2 norms of the derivatives of the adaptive parameters can be  decreased significantly. Two important points should be noted. One is that the bounds in \eqref{t1infb} are much tighter than those in \eqref{thetadotl2}, with terms of the form $\gamma^2/\ell$ no longer present. Finally, from Theorem \ref{thm:p}, it follows that the improved L-2 bounds in \eqref{t1infb} result in a reduced high frequency oscillations in the adaptive parameters.
\end{rem}

\begin{rem}
Noting the structure of the control input in \eqref{eqd:controller}, it follows directly that reduced oscillations in $\theta(t)$ and $k(t)$ results in reduced oscillation in the control input for the following reason. We note that\ben x_p(t)=\exp{a_mt}x_p(0)+\int_0^t \exp{a_m(t-\tau)}k_p\left( \tilde k r+\tilde \theta x_p\right) d \tau.\een Since $\tilde\theta$ and $\tilde k$ have reduced oscillations, $x_p(t)$ will be smooth, resulting in $\theta(t)x_p(t)$ and therefore $u(t)$ to have reduced oscillations. It should also be noted that with $\ell<\ell^*$ and $\gamma=\abs \ell$, it follows that $\norml{x(t)}\infty$ is independent of $\ell$. 
\end{rem}

\subsection{Simulation Studies for CRM}
Simulation studies are now presented to illustrate the improved transient behavior of the adaptive parameters and the peaking that can occur in the reference model. For these examples the reference system is chosen such that $a_m=-1,k_m=1$ and the plant is chosen as $a_p=1,k_p=2$. The adaptive parameters are initialized to be zero. Figures 1 through 3 are for an ORM adaptive system with the tuning gain chosen as $\gamma\in\{1,10,100\}$. Walking through Figures 1 through 3 it clear that as the tuning gain is increased the plant tracks the reference model more closely, at the cost of increased oscillations in the adaptive parameters. Then the CRM is introduced and the resulting responses are shown in Figures 4 through 6, for $\gamma=100$, and $\ell$=-10, -100, and -1000 respectively. First, it should be noted that no high frequency oscillations are present in these cases, and the trajectories are smooth, which corroborates  the inequalities \eqref{t1infb} in Theorem \ref{thm3}. As the ratio $\abs{\ell}/\gamma$ increases, as illustrated in Figures 4 through 6, the reference trajectory $x_m$ starts to deviate from the open-loop reference $x_m^o$, with the peaking phenomenon clearly visible in Figure 6 where $\abs \ell/\gamma =10$. This corroborates our results in section III as well.

\begin{figure}[h!]
\centering
\includegraphics[width=3.3in]{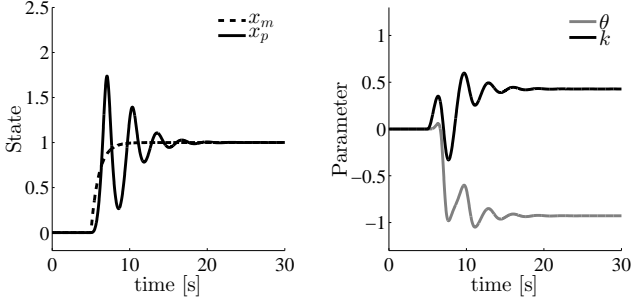}
\caption{Trajectories of the ORM adaptive system $\gamma=1$.}\label{fig:num1}
\end{figure}
\begin{figure}[h!]
\centering
\includegraphics[width=3.3in]{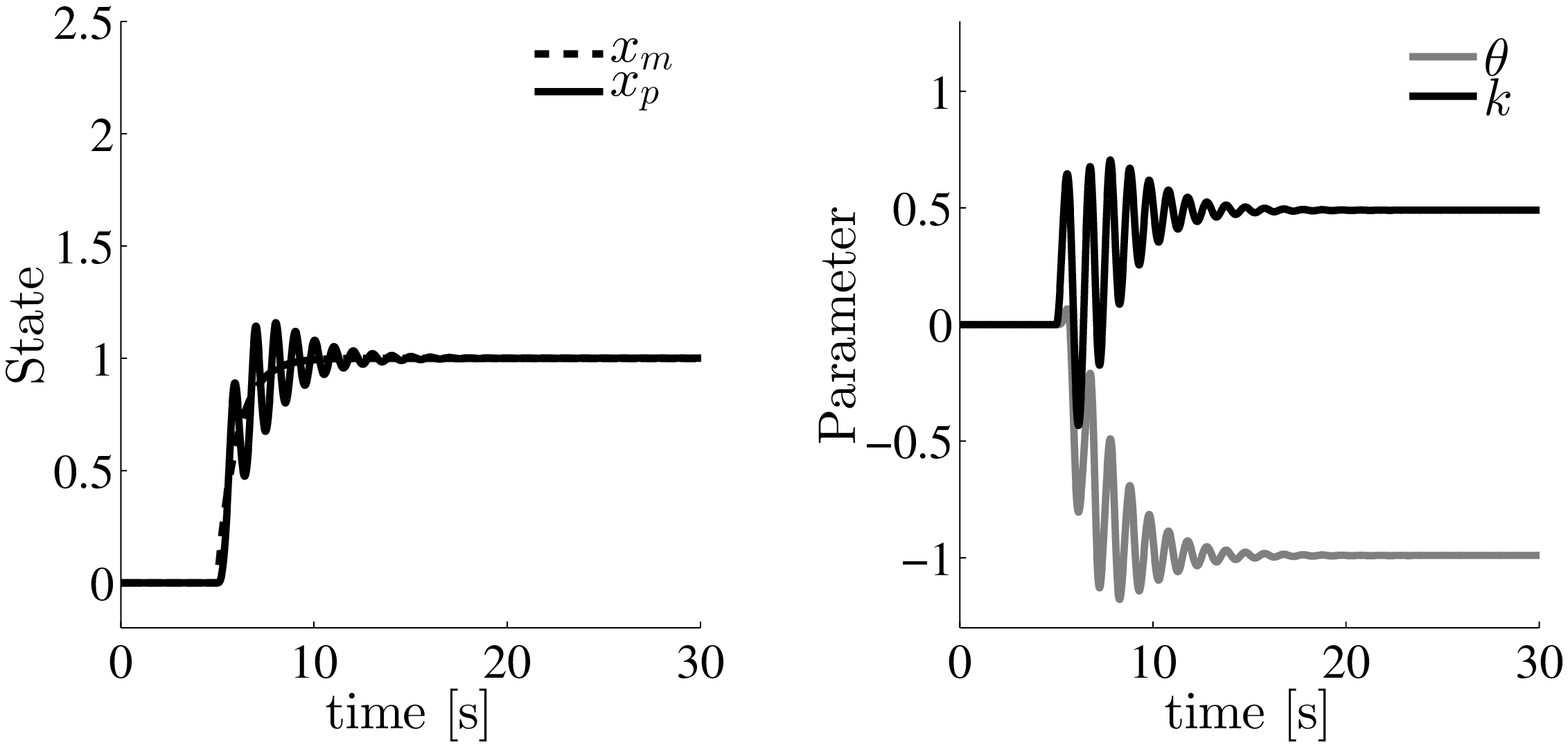}
\caption{Trajectories of the ORM adaptive system  $\gamma=10$.}\label{fig:num1}
\end{figure}
\begin{figure}[h!]
\centering
\includegraphics[width=3.3in]{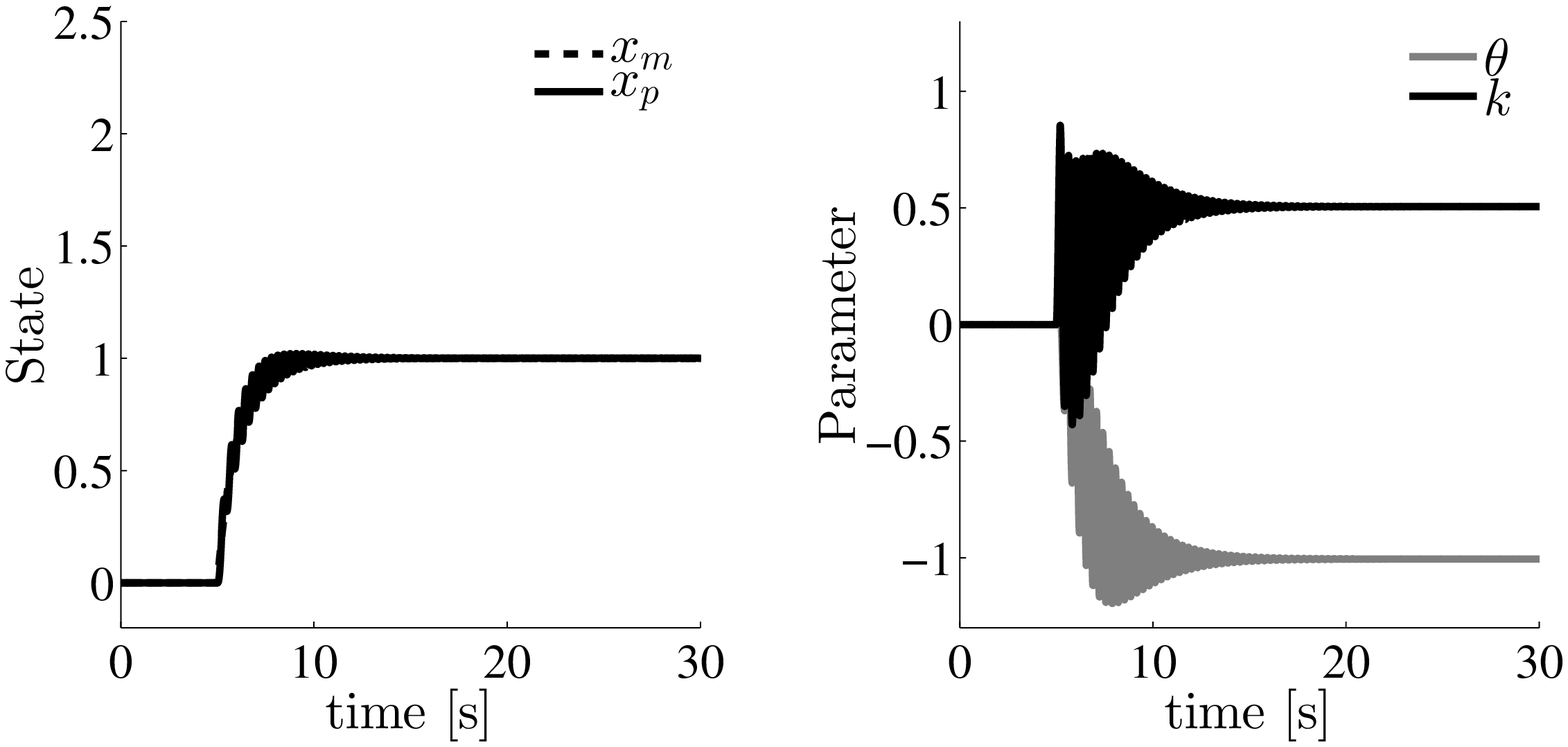}
\caption{Trajectories of the ORM adaptive system  $\gamma=100$.}\label{fig:num1}
\end{figure}
\begin{figure}[h!]
\centering
\includegraphics[width=3.3in]{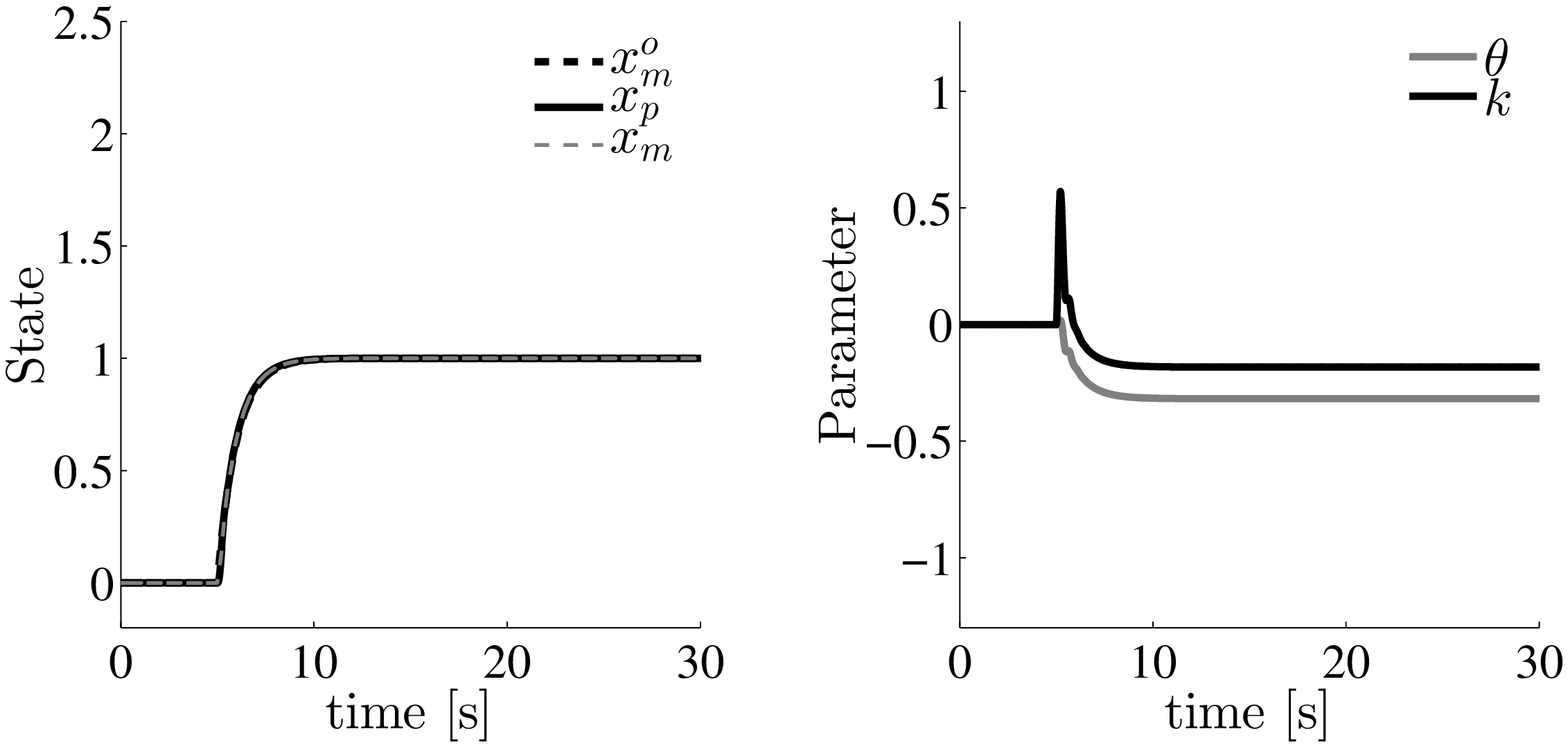}
\caption{Trajectories of the CRM adaptive system  $\gamma=100$, $\ell=-10$.}\label{fig:num1}
\end{figure}
\begin{figure}[h!]
\centering
\includegraphics[width=3.3in]{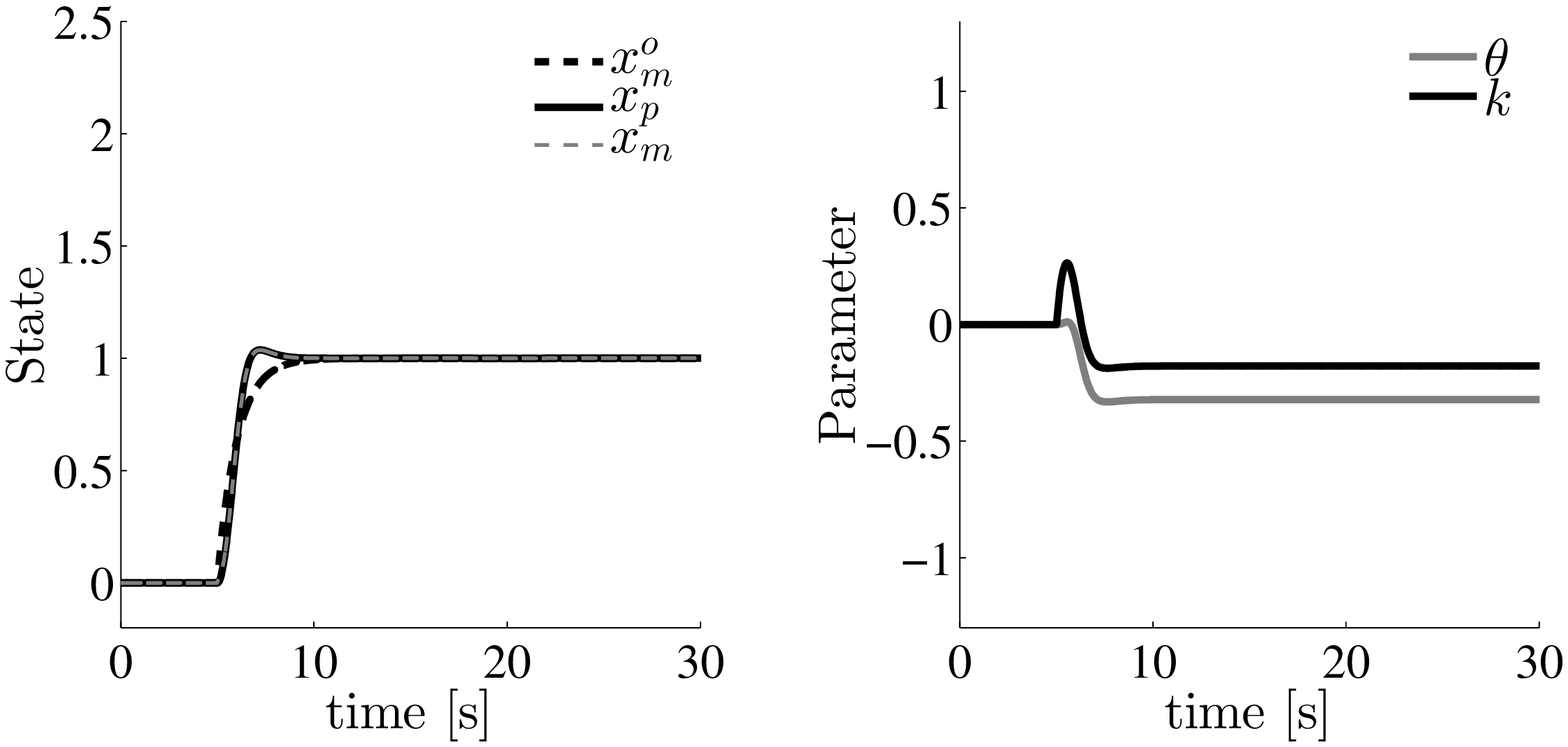}
\caption{Trajectories of the CRM adaptive system  $\gamma=100$, $\ell=-100$.}\label{fig:num1}
\end{figure}
\begin{figure}[h!]
\centering
\includegraphics[width=3.3in]{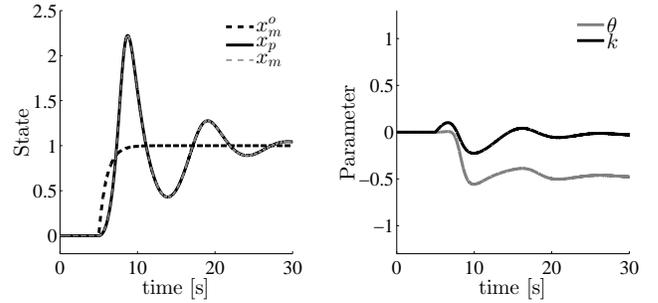}
\caption{Trajectories of the CRM adaptive system  $\gamma=100$, $\ell=-1000$.}\label{fig:num1}
\end{figure}

\section{CRM for States Accessible Control}\label{sec:ndim}
In this section we show that the same bounds shown previously easily extend to the states accessible case. Consider the $n$ dimensional linear system
\be\label{statesplant}
\dot x_p = Ax_p + B \Lambda u
\ee
with $B$ known, $A,\Lambda$ are unknown, and $\Lambda$. An a priori upper bound on $\Lambda$ is known and therefore we define
\ben\bar \lambda \triangleq \max_i \lambda_i(\Lambda),\een where $\lambda_i$ denotes the $i$-th Eigenvalue. The reference model is defined as 
\be\label{statesref}
\dot x_m = A_m x_m + B r - Le.
\ee
The control input is defined as 
\be\label{stu} u=\Theta x_p+ Kr.\ee It is assumed that there exists $\Theta^*$ and $K^*$ such that 
\ben\begin{split}
A+B\Lambda\Theta^*& = A_m \\
\Lambda K^* & =  I \end{split}
\een
and the parameter errors are then defined as $\tilde\Theta=\Theta-\Theta^*$ and  $\tilde K = K-K^*$. Defining the error as ${e=x_p-x_m}$, the update law for the adaptive parameters is then
\be\label{stateupdate}
\begin{split}
\dot\Theta& =\text{Proj}_{\Omega_1}(-\Gamma B^T P e x_p^T,\Theta) \\
\dot K&=\text{Proj}_{\Omega_2}(-\Gamma B^TPer^T,K)
\end{split}
\ee
where $P=P^T>0$ is the solution to the Lyapnov equation $(A_m+L)^TP+P(A_m+L)=-Q$ which exists for all $Q=Q^T>0$. With a slight abuse of notation the following definition is reused from the previous section,
\be
\sup_{\Theta,\Theta^*\in\Omega_1} \norm{\tilde{\Theta}}_F \triangleq \Theta_\text{max} \quad \text{and} \quad  \sup_{K,K^*\in\Omega_2} \norm{\tilde{K}}_F \triangleq K_\text{max},
\ee
where $\norm{\cdot}_F$ denotes the Frobenius norm. The adaptive system can be shown to be stable by using the following Lyapunov candidate,
\ben
V(t) = e^TPe + \text{Tr}(\Lambda\tilde\Theta^T\Gamma^{-1}\tilde\Theta) +  \text{Tr}(\Lambda\tilde K^T\Gamma^{-1}\tilde K)
\een
where after differentiating we have that $\dot V \leq -e^TQe$. We choose $L$  and $\Gamma$ in a special form to ease the analysis in the following sections.

\begin{assm}\label{ass1}
The free design parameters are chosen as 
\be\begin{split}
\Gamma & =  \gamma I_{n\times n}\\
L & =  - A_m + g I_{n\times n}
\end{split}\ee
where $\gamma>0$ and $g<0$.
\end{assm}
Assumption 1 allows us to choose a $P={1/2} I_{n\times n}$ in the Lyapunov equation and therefore $Q=-g{I_{n\times n}}$. Using these simplification the Lyapunov candidate derivative can be bounded as 
\be\label{vstate}
\dot V(t)\leq - \abs{g} \norm{e(t)}^2,\ee
and by direct integration we have
\be
\norm{e(t)}^2 \leq \frac{V(0)}{\abs{g}}.
\ee
Using the Gronwall-Bellman Lemma as was previously used in \eqref{gb1}-\eqref{ert}, we can deduce that
\be\label{gbstate}
V(t) \leq \frac{1}{2}\norm{e(0)}^2 \exp{-2\abs{g}t} +  \frac{\bar \lambda}{\gamma} \left(\Theta_\text{max}^2+K_\text{max}^2\right).
\ee
\begin{lem} \label{lem4} For all $\epsilon>0$ and $\delta>1$ there exists a $t_2$ such that the plant and reference model in \eqref{statesplant} and \eqref{statesref} respectively satisfy the bounds in \eqref{xm11} with $t_2$ replacing $t_1$.
\end{lem}

\begin{lem}\label{lem5}
Consider the adaptive system with the plant in \eqref{statesplant}, the controller in \eqref{stu}, the update law in \eqref{stateupdate}, the reference model as in \eqref{statesref} and $\Gamma$ and $L$ parameterized as in Assumption \ref{ass1}. Given a time $t_2\geq0$, there exists a $g^*$ s.t.
\be\label{extra2}
\sqrt{\int_{t_2}^\infty \norm{e}^2 d\tau} \leq  \frac{\norm{e(0)} }{\sqrt{2}\abs{g}} + \sqrt{\frac{\bar\lambda(\Theta_\text{max}^2+K_\text{max}^2)}{\gamma \abs{g}}}
\ee
for all $g\leq g^*$.
\end{lem}
\begin{proof} From \eqref{gbstate} we have that
\ben
V(t_2) \leq \frac{1}{2}\norm{e(0)}^2 \exp{-2\abs{g}t_2} +  \frac{\bar \lambda}{\gamma} \left(\Theta_\text{max}^2+K_\text{max}^2\right).\een
Using the above bound and integrating $-\dot V$ in \eqref{gbstate} from $t_2$ to $\infty$ and dividing by $\abs{g}$ leads to
\be
\int_{t_2}^\infty \norm{e}^2 d\tau \leq  \frac{\norm{e(0)}^2 \exp{-2\abs{g}t_2}}{2\abs{g}} +  \frac{\bar \lambda}{\gamma \abs{g}} \left(\Theta_\text{max}^2+K_\text{max}^2\right)
\ee
Taking the square root, noting that $\sqrt{\exp{-2 \abs{g}t_2}} =\exp{-\abs{g}t_2}$, and using the result from Lemma \ref{lem:exp}, we know that there exists an $g^*$ such that for all $g<g^*$, $\exp{-\abs{g}t_2}\leq \abs{g}^{-1/2}$.
\end{proof}

\begin{thm}Consider the adaptive system with the plant in \eqref{statesplant}, the controller in \eqref{stu}, the update law in \eqref{stateupdate}, the reference model as in \eqref{statesref}, $\Gamma$ and $L$ parameterized as in Assumption \ref{ass1}, with $t_2$ chosen as in Lemma \ref{lem4} and $g\leq g ^*$ where $g^*$ is given in Lemma \ref{lem5}. It can be shown that
\be\begin{split}\label{trav1}
\norm{x_m(t)}_{t \geq t_2}^2 \leq &\  c_5(t) + 2 \left( \frac{\norm{A_m}^2} {\abs{g}^2}+1\right) {\frac{a_1}{a_2}} \norm{e(0)}^2 \\ & +4 \bar \lambda \left( \frac{\norm{A_m}^2} { \gamma \abs{g}^2}+\frac{\abs g}{\gamma}\right) {\frac{a_1}{a_2}}\left(\Theta_\text{max}^2+K_\text{max}^2\right)
\end{split}\ee
where
\ben\begin{split}
c_5 & \triangleq  2\left(\norml{x_m^o}\infty + \norm{ x_m(t_2)}a_1\exp{-a_2(t-t_2)} \right)^2 \\
\exp{A_mt} & \leq a_1\exp{a_2t}
\end{split}
\een
with $a_1,a_2>0$.
\end{thm}
\begin{proof}
The existence of $a_1,a_2>0$ such that ${\exp{A_mt}  \leq a_1\exp{-a_2t}}$ follows from the fact that $A_m$ is Hurwtiz \cite{mol03}. Consider the dynamical system in \eqref{statesref} for ${t \geq t_2}$,
\ben\begin{split}
\norm{x_m(t)}_{t\geq t_2} \leq & \norml{x_m^o}\infty + \norm{ x_m(t_1)}a_1\exp{-a_2(t-t_1)} \\&+ \norm{L} a_1\int_{t_2}^\infty \exp{-a_2(t-\tau)} \norm{e(\tau)} d\tau.
\end{split}\een
Using Cauchy Schwartz along with \eqref{extra2} in Lemma \ref{lem5}
we have that
\ben\begin{split}
\norm{ x_m(t)}_{t\geq t_2} \leq & \norml{x_m^o}\infty+ \norm{\ x_m(t_2)}\exp{-\abs{a_m}(t-t_1)} \\&+ \frac{\norm L \sqrt{a_1}}{\sqrt{2{a_2}}}  \left(\frac{\norm{e(0)}}{\sqrt 2 \abs{g}}+\sqrt{\frac{\bar \lambda(\Theta_\text{max}^2+K_\text{max}^2)}{\gamma\abs{g}}}\right).
\end{split}\een
Squaring and using the fact that $L=-A_m + g I_{n\times n}$ and thus $\norm{L}\leq \norm{A_m}+ g$ we have that
\ben\begin{split}
\norm{x_m(t)}_{t \geq t_2}^2 \leq &\  c_5(t) + \frac{(\norm{A_m}+g)^2} {\abs{g}^2} {\frac{a_1}{a_2}} \norm{e(0)}^2 \\ & +\frac{ 2(\norm{A_m}+g)^2\bar \lambda }{\gamma\abs{g}} {\frac{a_1}{a_2}}\left(\Theta_\text{max}^2+K_\text{max}^2\right).
\end{split}\een
Inequality \eqref{trav1} follows since $(\norm{A_m}+\abs g)^2 \leq 2\norm{A_m}^2 +2 \abs{g}^2$.
\end{proof}

\begin{rem}
Just as in the scalar case, we have derived a bound  for $\norm{x_m(t)}^2_{t\geq t_2}$ which is once again a function of the initial condition of the plant and controller, but also dependent on a component which is proportional to ${\abs{g}}/\gamma$. Therefore, by choosing $\frac{\abs g}{\gamma}=1$ with $\gamma>0$ we can have bounded peaking in the reference model. \end{rem}

\begin{thm}\label{thm:states}
The adaptive system with the plant in \eqref{statesplant}, the controller in \eqref{stu}, the update law in \eqref{stateupdate}, the reference model as in \eqref{statesref}, $\Gamma$ and $L$ parameterized as in Assumption \ref{ass1}, with $t_2$ chosen as in Lemma \ref{lem4}, $g\leq g ^*$ where $g^*$ is given in Lemma \ref{lem5} and $\gamma$ chosen such that $\gamma=\abs{g}$ the following bounds are satisfied for all $\gamma \geq 1$:
\be\label{}
\begin{split}
\int_{t_2}^\infty \norm{\dot K}^2 d\tau\leq &  \norm{B}\left(\norm{e(0)}^2+K_\text{max}^2+\Theta_\text{max}^2\right) \norml{r(t)}\infty \\
\int_{t_2}^\infty \norm{\dot \Theta}^2 d\tau \leq &   \norm{B}\left(\norm{e(0)}^2+\Theta_\text{max}^2+K_\text{max}^2\right) \cdot  \\ & \quad \left(c_6+\frac{c_7}{g^2}+\frac{c_8}{\sqrt{\abs g}} + \frac{c_9}{\gamma}\right) \end{split} \ee  
where $c_6,c_7,c_8,c_9$ are independent of $\gamma$ and $g$, and are only a function of the initial conditions of the system and the fixed design parameters.
\end{thm}

\begin{proof}
The proof follows the same steps as used to derive the bounds in Theorem \ref{thm3}.\end{proof}

\begin{rem}
It should be noted that if $\gamma$ and $\abs{g}$ are increased while holding $\gamma=\abs{g}$, the L-2 norms of the derivatives of the adaptive parameters can be  decreased significantly. \end{rem}

\begin{rem} The similarity of the bounds in Theorem \ref{thm:states} to those in Theorem \ref{thm3} implies that the same bounds on frequencies and corresponding amplitudes of the overall adaptive systems as in Theorem \ref{thm:p} hold here in the higher-order plant as well. \end{rem}

We note that robustness issues have not been addressed with the CRM architecture in this work. However, recent results in \cite{mat12,mat13,hus13c1,hus13c2} have shown that adaptive systems do have a time-delay margin and robustness to unmodeled dynamics when projection is used in the update law. While we expect similar results to hold with CRM as well, a detailed investigation of the same as well as comparisons of their robustness properties to their ORM counterparts are topics for further research.

\section{CRM Composite Control with Observer Feedback}
\label{cmracco}
In this section, we show that the tools introduced to demonstrate smooth transient in CRM-adaptive systems can be used to analyze CMRAC systems introduced in \cite{duarte1989combined,slotine1989composite,eugeneTAC09}. As mentioned in the introduction, these systems were observed to exhibit smooth transient response, and yet no analytical explanations have been provided until now for this behavior. Our focus is on first-order plants for the sake of simplicity. Similar to Section \ref{sec:ndim}, all results derived here can be directly extended to higher order plants whose states are accessible.

The CMRAC system that we discuss in this paper differs from that in \cite{duarte1989combined} and includes an observer whose state is fed back for control rather than the plant state. As mentioned in the introduction, we denote this class of systems as CMRAC-CO and is described by the plant in \eqref{eqd:plant}, the reference model in \eqref{eqd:reference}, an observer as
\be\label{cxo}
\dot x_o(t) = \ell (x_o - x_p) + (a_m-k_p\hat\theta) x_o(t) + k_p u(t),
\ee
and the control input is given by 
\be\label{ccontrol}
u= \theta x_o + k^* r.
\ee
In the above $k_p$ is assumed to be known for ease of exposition. The feedback gain $\ell$ is chosen so that 
\be\label{gt}
g_\theta \triangleq a_m+\ell+\abs{k_p\theta^*}<0.
\ee

Defining ${e_m = x_p-x_m}$ and ${e_o=x_o-x_p}$, the error dynamics are now given by
\be
\begin{split}\label{}
\dot e_m(t) & =  (a_m+\ell) e_m +k_p \tilde \theta x_o +k_p\theta^*e_o \\
\dot e_o(t) & =  (a_m+\ell) e_o - k_p \bar\theta x_o. 
\end{split}
\ee
where ${\tilde\theta=\theta-\theta^*}$ and  ${\bar\theta=\hat \theta -\theta^*}$
with $\theta^*$ satisfying ${a_p + k_p\theta^* = a_m}$ and ${k_p k^*= k_m}$.
The update laws for the adaptive parameters are then defined with the update law
\be\label{cupdate}\begin{split}
\dot \theta & = \text{Proj}_{\Omega}  (-\gamma\text{sgn}{(k_p)} e_{m} x_o,\theta)  - \eta\epsilon_\theta\\
\dot {\hat \theta} & = \text{Proj}_{\Omega} ( \gamma\text{sgn}(k_p) e_o x_o,\hat \theta)+\eta  \epsilon_\theta\\
\epsilon_\theta & = \theta-\hat\theta
\end{split}\ee
where $\gamma,\eta>0$ are free design parameters. As before we define the bounded set
\be
 \Theta_\text{max}\triangleq \text{max} \left\{ \sup_{\theta,\theta^*\in\Omega} \norm{\tilde{\theta}}, \sup_{\hat\theta,\theta^*\in\Omega} \norm{\bar{\theta}}  \right\}. 
\ee
We first establish stability and then discuss  the improved transient response.

\subsection{Stability}
The stability of the CMRAC-CO adaptive system given by \eqref{eqd:plant}, \eqref{eqd:reference}, \eqref{cxo}-\eqref{cupdate} can be verified with the following Lyapunov candidate
\be\label{clyap}
 V(t) = \frac{1}{2}\left(e_m^2 + e_o^2 + \frac{\abs{k_p}}{\gamma} \tilde\theta^2 + \frac{\abs{k_p}}{\gamma}\bar\theta^2 \right)
\ee
which has the following derivative
\be\label{vdc}
\dot V \leq g_\theta e_m^2 + g_\theta e_o^2 -\frac{\eta \abs{k_p}}{\gamma} \epsilon_\theta^2.
\ee
Boundedness of all signals in the system follows since ${g_\theta<0}$. 
From the integration of \eqref{vdc} we have $\left\{e_m,e_o,\epsilon_\theta \right\} \in \mathcal L_\infty \cap \mathcal L_2$ and thus $\lim_{t \to \infty}\{e_m,e_o,\epsilon_\theta\} =\{0,0,0\}$. Using the Gronwall-Bellman Lemma as was previously used in \eqref{gb1}-\eqref{ert}, we can deduce that
\be\label{gbstate}
V(t) \leq \frac{1}{2}\left({e_m(0)}^2+e_o(0)^2\right) \exp{-2\abs{g_\theta}t} +  \frac{\abs{k_p}}{\gamma}\theta_\text{max}^2.
\ee
It should be noted that the presence of a non-zero $\ell$ is crucial for stability, as $g_\theta$ cannot be guaranteed to be negative if $\ell=0$.

\subsection{Transient performance of CMRAC-CO}
Similar to Sections II and III we divide the timeline into $[0,t_3]$ and $[t_3,\infty)$, where $t_3$ is arbitrarily small. We first derive bounds for the system states over the initial $[0,t_3]$ in Lemma \ref{lem6}, bounds for the tracking, observer, parameter estimation errors $e_m$, $e_o$ and $\epsilon_\theta$ over $[t_3,\infty)$ in Lemma \ref{lem7}, bounds for $x_o$ over $[t_3,\infty)$ in Theorem \ref{xob}, and finally bounds for the parameter derivatives $\dot\theta$ and $\dot{\hat\theta}$ in Theorem \ref{thm8}.
\begin{lem}\label{lem6}
Consider the CMRAC-CO adaptive system with the plant in \eqref{eqd:plant}, with the controller defined by \eqref{ccontrol}, the update law in \eqref{cupdate} and with the reference model as in \eqref{eqd:reference}. For all $\delta>1$ and $\epsilon>0$, there exists a time $t_3\geq0$ such that
\be\begin{split}\label{}
\norm{x_p(t)} &\leq \delta \norm{x_p(0)} + \epsilon \left(\norml{r(t)}\infty  +\sqrt{2V(0)}\right) \\
\norm{x_o(t)}& \leq \delta \norm{x_p(0)} + \epsilon \norml{r(t)}\infty  + (1+\epsilon) \sqrt{2V(0)}\\ 
\end{split}\ee
$\forall\   0\leq t\leq t_3$.
\end{lem}
\begin{proof}
The plant in \eqref{eqd:reference} with the controller in \eqref{ccontrol} can be represented as
\ben
\dot x_p = (a_p +k_p\theta) x_p + k_p  \left(r +\theta e_o\right)
\een
where we note that $(a_p +k_p\theta) $ can be positive. This leads to the inequality
\ben\begin{split}
\norm{x_p(t)} \leq &\norm{x_p(0)} \exp{(a_p +\abs{k_p}\Theta_\text{max})t}  + \int_0^t\exp{(a_p +\abs{k_p}\Theta_\text{max})(t-\tau) } \cdot \\ & \hspace{.6 in}\abs{k_p}\left(\Theta_\text{max}\norm{r(\tau)}+\Theta_\text{max} \norm{e_o(\tau)}\right) d\tau.
\end{split}\een
For any ${\delta>}1$ and any ${\epsilon>0}$, it follows from the above inequality that a $t_3$ exists such that 
$\exp{(a_p +\abs{k_p}\Theta_\text{max})t}\leq \delta$ and ${\int_0^t\exp{(a_p +\abs{k_p}\Theta_\text{max})(t-\tau) }\abs{k_p}\Theta_\text{max} d\tau}\leq \epsilon$, ${0\leq t \leq t_3}$ given $\delta>0$ and $\epsilon>0$. From the structure of the Lyapunov candidate in \eqref{clyap} and the fact that $\dot V\leq 0$ we have that $\norml{e_o(t)}\infty\leq \sqrt{2V(0)}$. The bound on $x_o(t)$ follows from the fact that $\norm{x_o}\leq \norm{x_p}+\norm{e_o}$.
\end{proof}

\begin{lem}\label{lem7}
Consider the adaptive system with the plant in \eqref{eqd:plant}, the controller in \eqref{ccontrol}, the update law in \eqref{cupdate}, and the reference model as in \eqref{eqd:reference}. Given a time $t_2\geq0$, there exists a $g_\theta^*$ s.t.
\be\label{ceb}\begin{split}
\sqrt{\int_{t_3}^\infty {e_m}^2 d\tau} &\leq  \frac{\sqrt{e_m(0)^2+ e_o(0)^2} }{\sqrt{2}\abs{g_\theta}} + \sqrt{\frac{\abs{k_p}}{\gamma \abs{g_\theta}}}\Theta_\text{max}\\
\sqrt{\int_{t_3}^\infty {e_o}^2 d\tau} &\leq \frac{\sqrt{e_m(0)^2+ e_o(0)^2 }}{\sqrt{2}\abs{g_\theta}} + \sqrt{\frac{\abs{k_p}}{\gamma \abs{g_\theta}}}\Theta_\text{max}\\
\sqrt{\int_{t_3}^\infty {\epsilon_\theta}^2 d\tau} &\leq \frac{\sqrt{\gamma}\sqrt{e_m(0)^2+ e_o(0)^2 }}{\sqrt{2\eta \abs{k_p}\abs{g_\theta}}} + \sqrt{\frac{1}{\eta}}\Theta_\text{max}\\
\end{split}
\ee
for all $g_\theta\leq g_\theta^*$.
\end{lem}

\begin{thm}\label{xob} Consider the adaptive system with the plant in \eqref{eqd:plant}, the controller in \eqref{ccontrol}, the update law in \eqref{cupdate}, the reference model as in \eqref{eqd:reference}, with $t_3$ chosen as in Lemma \ref{lem6} and $g_\theta\leq g_\theta^*$ where $g_\theta^*$ is given in Lemma \ref{lem7}. It can be shown that
\be\begin{split}
\norm{x_o(t)}_{t \geq t_3}^2 \leq &\  c_{10}(t) +\frac{\abs{\ell}^2} {\abs{g_\theta}^2} 2 \sqrt{a_4}\left({e_m(0)^2+ e_o(0)^2 }\right)\\ & + {\frac{\abs{\ell}^2\abs{k_p}}{\gamma \abs{g_\theta}}}4a_4\Theta_\text{max}^2
\end{split}\ee
where
\ben\begin{split}
c_{10} & \triangleq  2\left(a_3\abs{k_p}\norml{r}\infty + \norm{ x_o(t_3)}\exp{a_\theta(t-t_3)} \right)^2 \\
a_\theta &\triangleq a_m+{k_p}\epsilon_\theta \\
\end{split}
\een
and
\ben\begin{split}
\int_{t_3}^\infty \exp{a_\theta(t-\tau)}d\tau & \leq a_3 \\
{\int_{t_3}^\infty \exp{2a_\theta(t-\tau)} d\tau} & \leq a_4
\end{split}
\een
with $0\leq a_i<\infty,\ i\in\{3,4\}$.
\end{thm}
\begin{proof} Given that $\lim_{t\to\infty} \epsilon_\theta(t)=0$ we have from \eqref{gt} that $\lim_{t\to\infty}a_\theta=a_m$. Thus $\lim_{t\to\infty}\exp{a_\theta t}=0$. Therefore, $a_3,a_4<\infty$. Consider the dynamical system in \eqref{cxo} for $t \geq t_3$.
\ben\begin{split}
\norm{x_o(t)}_{t\geq t_3} \leq &  \norm{ x_o(t_3)}\exp{a_\theta(t-t_3)}  \\
 & +\int_{t_3}^\infty \exp{a_\theta(t-\tau)} \left(\abs{\ell}\norm{e_o(\tau)} +\abs{k_p}\norm{r(\tau)}\right)d\tau.
\end{split}\een
Using Cauchy Schwartz and Lemma \ref{lem7} as before
we have 
\ben\begin{split}
\norm{ x_o(t)}_{t\geq t_3} \leq & a_3\abs{k_p}\norml{r}\infty+ \norm{\ x_o(t_3)}\exp{-\abs{a_m}(t-t_3)} \\&+ \abs{\ell} \sqrt{a_4} \frac{\sqrt{e_m(0)^2+ e_o(0)^2 }}{\sqrt{2}\abs{g_\theta}} \\ &+\abs{\ell} \sqrt{a_4}a_4  \sqrt{\frac{\abs{k_p}}{\gamma \abs{g_\theta}}}\Theta_\text{max}.
\end{split}\een
Squaring and using the inequality $(a+b)^2\leq 2a^2 + 2b^2$ twice, we have our result.
\end{proof}

\begin{cor} For the system presented in Theorem \ref{xob} setting $\gamma = \abs{g_\theta}$ and taking the limit as $\ell \to-\infty$ the following bound holds for $x_o(t)$
\be\begin{split}
\lim_{\ell\to-\infty} \norm{x_o(t)}_{t \geq t_3}^2 \leq &\  c_{10}(t) + 2 \sqrt{a_4}\left({e_m(0)^2+ e_o(0)^2 }\right)\\ & + \abs{k_p} 4a_4\Theta_\text{max}^2.
\end{split}
\ee 
\end{cor}

\begin{thm}\label{thm8}
The adaptive system with the plant in \eqref{eqd:plant}, the controller defined by \eqref{ccontrol}, the update law in \eqref{cupdate} with the reference model as in \eqref{eqd:reference}, with $t_3$ chosen as in Lemma \ref{lem6} and $g_\theta\leq g_\theta^*$ where $g_\theta^*$ is given in Lemma \ref{lem7} and $\gamma$ chosen such that $\gamma=\abs{g_\theta}$ the following bounds are satisfied for all $\gamma \geq 1$:
\be\label{}
\int_{t_3}^\infty \norm{\dot \theta}^2 d\tau \leq \alpha \text{ and } \int_{t_3}^\infty \norm{\dot {\hat \theta}}^2 d\tau \leq \alpha
\ee  
with 
\ben
\alpha \triangleq \left(\frac{\ell^2}{g_\theta^2}+\frac{\eta}{\abs{k_p}}\right)\left(e_m(0)^2+ e_o(0)^2 + \abs{k_p}\Theta_\text{max}^2\right)c_{11} 
\een
where $c_{11}$ is independent of $\gamma$ and $g_\theta$, and is only a function of the initial conditions of the system and the fixed design parameters.\end{thm}

\begin{rem}
Note that \ben\lim_{\ell\to-\infty}\frac{\ell^2}{g_\theta^2}=1.\een Thus for large $\abs \ell$ the truncated L-2 norm of $\dot\theta$ is simply a function of the initial conditions of the system and the tuning parameter $\eta$.
\end{rem}

\begin{rem} The similarity of the bounds in Theorem \ref{thm8} to those in Theorem \ref{thm3} implies that the same bounds on frequencies and corresponding amplitudes of the overall adaptive system hold here in the CMRAC-CO case as well.\end{rem}

\subsection{Robustness of CMRAC--CO to Noise }
As mentioned earlier, the benefits of the CMRAC--CO is the use of the observer state $x_o$ rather than the actual plant state $x$. This implies that the effect of any measurement noise on system performance can be reduced. This is explored in this section and Section \ref{cmracsim}. 

Suppose that the actual plant dynamics is modified from \eqref{eqd:plant} as
\be\label{p:n}
\dot x_a(t)= a_p x_a(t)+k_p u(t), \quad x_p(t)=x_a(t)+n(t)
\ee
where $n(t)$ represents a time varying disturbance. For ease of exposition, we assume that $n\in \mathcal C_1$.

This leads to a set of modified error equations
\be
\begin{split}\label{cne}
\dot e_m(t) =& (a_m+\ell) e_m +k_p \tilde \theta (t)x_o +k_p\theta^*e_o  + \xi(t) \\
\dot e_o(t) =& (a_m+\ell-k_p\theta^*) e_o -k_p \bar \theta(t)x_o -\xi(t)
\end{split}
\ee
where
\be
\xi(t) \triangleq \dot \eta(t) - a_p \eta(t)
\ee
\begin{thm}
The adaptive system with the plant in \eqref{p:n}, the controller defined by \eqref{ccontrol}, the update law in \eqref{cupdate} with the reference model as in \eqref{eqd:reference}, and $\ell$ chosen such that $a_m+\ell+2\abs{k_p}\abs{\theta^*}<0$, all trajectories are bounded and 
\be\begin{split}\label{eq:59}
V(t) \leq & \frac{1}{2}\left({e_m(0)}^2+e_o(0)^2\right) \exp{-2\abs{g_n}t} \\ &+  \frac{\abs{k_p}}{\gamma}\Theta_\text{max}^2+ \frac{1}{4\abs{g_n}^2} \norml{\xi(t)}\infty^2.
\end{split}\ee
where 
\be g_n \triangleq a_m+\ell + 2\abs{k_p}\abs{\theta^*}.\ee
\end{thm}
\begin{proof}
Taking the time derivative of $V$ in \eqref{clyap} results in \be
\begin{split}\dot V \leq&   g_n \left(\norm{e_m}^2+\norm{e_o}^2\right) -\abs{k_p}\frac\eta\gamma\epsilon_\theta^2 \\&+ \norm{\xi(t)}\norm{e_m(t)} + \norm{\xi(t)}\norm{e_o(t)}.
\end{split}\ee
completing the square in $\norm{e_m}\norm n$ and $\norm{e_o} \norm{n}$
\ben\begin{split}
\dot V \leq & -\abs{g_n}/2 \left(\norm{e_m}^2 +\norm{e_o}^2\right) -\abs{k_p}\frac\eta\gamma\epsilon_\theta^2 \\ 
&-\abs{g_n}/2\left(\norm{e_m} -  1/ \abs{g_n} \norm{\xi(t)}\right)^2 \\  
&-\abs{g_n}/2 \left(\norm{e_o} -  1/ \abs{g_n}  \norm{\xi(t)}\right)^2 \\
& + 1/(4\abs{g_n}) \norm{\xi(t)}^2.
\end{split}
\een
Neglecting the negative terms in lines 2 and 3 from above and the term involving $\epsilon_\theta$ we have that
\ben
\dot V \leq  -\abs{g_n}/2 \left(\norm{e_m}^2 +\norm{e_o}^2\right) + 1/(4\abs{g_n}) \norm{\xi(t)}^2,
\een
and in terms of $V$ gives us
\ben
\dot V \leq -\abs{g_n} V +\frac{1}{2}  \frac{\abs{k_p} \abs{g_n}}{\gamma}\left(\tilde\theta^2+\bar\theta^2\right)+ \frac{1}{4\abs{g_n}}\norm{\xi(t)}^2.
\een
Using the Gronwall-Bellman Lemma and substitution of $V(t)$ leads to the bound in \eqref{eq:59}.
\end{proof}

\subsection{Simulation Study} \label{cmracsim}
For this study a scalar system in the presence of noise is to be controlled with dynamics as presented in \eqref{p:n}, where  $n(t)$ is a deterministic signal used to represent sensor noise. $n(t)$ is generated from a Gausian distribution with variance 1 and covariance 0.01, deterministically sampled using a fixed seed at 100 Hz, and then passed through a saturation function with upper and lower bounds of 0.1 and -0.1 respectively. For the CMRAC-CO systems the reference model is chosen as \eqref{eqd:reference} with the rest of the controller described by \eqref{cxo}-\eqref{cupdate}. 

The CMRAC system used for comparison is identical to that in \cite{duarte1989combined}. For CMRAC the reference dynamics are now chosen as $x_m^o$  in \eqref{orm}, the observer is the same as CMRAC-CO \eqref{cxo}. 

Further differences arise with the control law being chosen as
\ben
u=\theta x_p + k^*r
\een 
The open-loop error $e^o=x_p-x_m^o$ updates the direct adaptive component, with the regressor becoming $x_p$ instead of $x_o$ for both $\theta$ and $\hat\theta$ update laws:
\be\begin{split}\label{arse}
\dot \theta & = \text{Proj}_{\Omega}  (-\gamma\text{sgn}{(k_p)} e_{m}^o x_p,\theta)  - \eta\epsilon_\theta\\
\dot {\hat \theta} & = \text{Proj}_{\Omega} ( \gamma\text{sgn}(k_p) e_o x_p,\hat \theta)+\eta  \epsilon_\theta.
\end{split}\ee
The complete CMRAC and CMRAC-CO systems are given in Table I with the design parameters given in Table II.


\begin{figure}[t]
\centering
\psfrag{e}[cc][cc][.9]{$e$}
\psfrag{x}[cc][cc][.9]{$x$}
\psfrag{t}[cc][cc][.8]{$t$}
\psfrag{r}[cc][cc][.9]{$x_m$}
\psfrag{openloop}[cl][cl][.8]{open--loop}
\psfrag{closedloop}[cl][cl][.8]{closed--loop}
\psfrag{closedloope0}[cl][cl][.8]{$e^o$: closed--loop}
\psfrag{c1}[cc][cl][.8]{Region 1}

\psfrag{c2}[cc][cl][.8]{Region 2}

\psfrag{c3}[cc][cl][.8]{Region 3}
\includegraphics[width=3.4in]{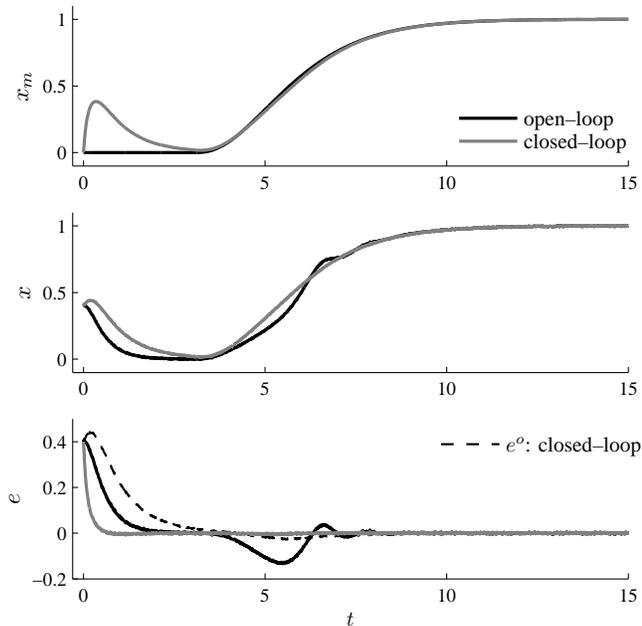}
\caption{(top) reference model trajectories $x_m$, (middle) state $x$, and (bottom)
model following $e$.}\label{fig:ob1}
\end{figure}

\begin{figure}[h]
\centering
\psfrag{u}[cc][cc][.9]{$u$}
\psfrag{u1}[cc][cc][.9]{$\frac{\Delta u}{\Delta t}$}
\psfrag{t}[cc][cc][.8]{$t$}
\psfrag{r}[cc][cc][.8]{$x_m$}
\psfrag{h1}[cc][cc][.8]{$\theta$}

\psfrag{h2}[cc][cc][.8]{$\hat\theta$}
\psfrag{c1}[cc][cl][.8]{Region 1}

\psfrag{c2}[cc][cl][.8]{Region 2}

\psfrag{c3}[cc][cl][.8]{Region 3}

\psfrag{openloop}[cl][cl][.8]{open--loop}
\psfrag{thetastar}[cl][cl][.8]{$\theta^*(t)$}
\psfrag{closedloop}[cl][cl][.8]{closed--loop}
\includegraphics[width=3.4in]{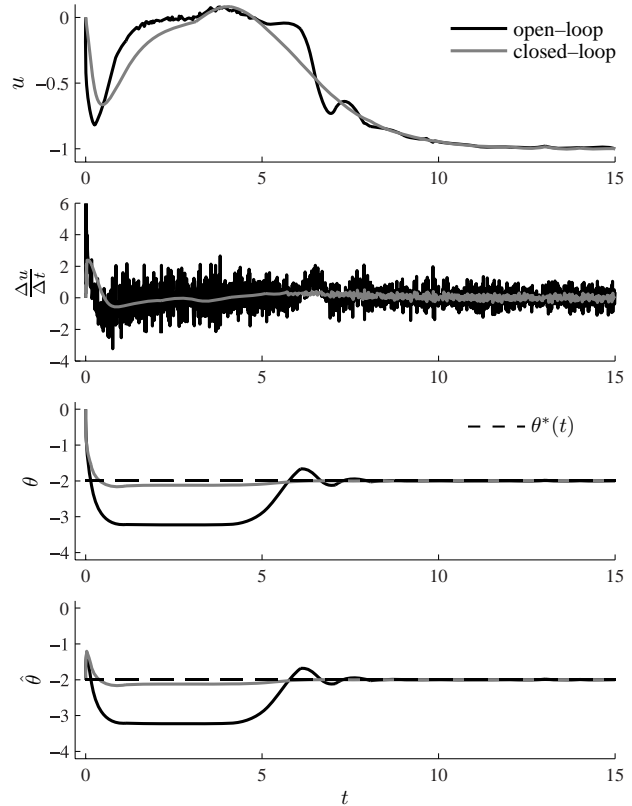}
\caption{(top) Control input $u$, (middle--top) discrete rate of change of control input $\Delta u/ \Delta t$, (middle--bottom) adaptive parameter $\theta(t)$ and (bottom) adaptive parameter $\hat\theta(t)$.}\label{fig:ob2}
\end{figure}

\begin{table}[h]
\caption{Test Case Equations} \label{baus}
\centering
\begin{tabular}{l c c}
  Paramater & CMRAC & CMRAC-CO  \\ \hline
  Reference Dynamics  & Equation \eqref{orm} & Equation \eqref{eqd:reference} \\
  Observer Dynamics & Equation \eqref{cxo} & Equation \eqref{cxo} \\
  Reference Error & $e_m^o=x_p-x_m^o$ & $e_m = x_p-x_m$ \\
  Observer Error & $e_o=x_o-x_p$ & $e_o = x_o-x_p$ \\
  Input & $u=\theta x_p + r$ & $u=\theta x_o + r$ \\
  Update Laws & Equation \eqref{arse}  & Equation \eqref{cupdate} \\ \hline
\end{tabular}  \end{table}

\begin{table}[h]
\caption{Simulation Parameters} \label{baus}
\centering
\begin{tabular}{l c c}
  Paramater & Value  \\ \hline
  $a_p$ & 1 \\
  $k_p$ & 1  \\
    $a_m$ & -1  \\
    $k_m$ & 1  \\
    $\ell$ & -10 \\
    $\gamma$ & 100  \\
    $\eta$ & 1  \\ \hline
\end{tabular}  \end{table}

The simulations have two distinct regions of interest, with Region 1 denoting the first 4 seconds, Region 2 denoting the 4 sec to 15 sec range. In Region 1, the adaptive system is subjected to non--zero initial conditions in the state and the reference input is zero. At $t=4$ sec, the beginning of Region 2, a filtered step input is introduced. Figures \ref{fig:ob1} and \ref{fig:ob2} illustrate the response of the CMRAC--CO adaptive system over 0 to 15 seconds, with $x_m$, $x$, and $e_m$ indicated in Figure \ref{fig:ob1}, and $u$, $\Delta u/ \Delta t$, $\theta$ and $\hat\theta$ indicated in Figure \ref{fig:ob2}. The addition of sensor noise makes the output $x_p$ not differentiable and therefore we use the discrete difference function $\Delta$ to obtain the discrete time derivative of the control input, where
\ben
\frac{\Delta u}{\Delta t} \triangleq \frac{u(t_{i+1})-u(t_i)}{t_{i+1}-t_i}, \qquad t_{i+1}-t_i = 0.01.
\een

In both cases, the resulting performance is compared with the classical CMRAC system. The first point that should be noted is a satisfactory behavior in the steady-state of the CMRAC--CO adaptive controller.  We note a significant difference between the responses of CMRAC--CO and CMRAC systems, which pertains to the use of filtered regressors in CMRAC--CO. An examination of $\Delta u/ \Delta t$ in Figure \ref{fig:ob2} clearly illustrates the advantage of CMRAC--CO.

\subsection{Comments on CMRAC and CMRAC--CO}\label{sec:talker}

As discussed in the Introduction, combining indirect and direct adaptive control has always been observed to produce desirable transient response in adaptive control. 
While the above analysis does not directly support the observed transient improvements with CMRAC, we provide a few speculations below: The free design parameter $\ell$ in the identifier is typically chosen to have eigenvalues faster than the plant that is being controlled. Therefore the identification model following error $e_i$ converges rapidly and $\hat\theta(t)$ will have smooth transients. It can be argued that the desirable transient properties of the identifier pass on to the direct component through the tuning law, and in particular $\epsilon_\theta$.

\section{CRMs in other Adaptive Systems}

While CRMs can be traced to \cite{lee97} in the context of direct model reference adaptive control, such a closed loop structure has always been present in, adaptive observers, tuning function designs, and in a similar fashion in adaptive control of robots. These are briefly described in the following sections.

\subsection{Adaptive Backstepping with Tuning Functions}
The control structure presented here is identical to that presented in \cite[\S 4.3]{kkkbook}. Consider the unknown system
\be
\begin{split}
\dot x_1 & = x_2 + \varphi_1(x_1)^T\theta^*\\
\dot x_1 & = x_3 + \varphi_2(x_1,x_2)^T\theta^*\\
&\ \, \vdots \\
\dot x_{n-1} & = x_n + \varphi_{n-1}(x_1,\ldots,x_{n-1})^T\theta^*\\
\dot x_n & = \beta(x) u  + \varphi_n(x)^T\theta^*
\end{split}
\ee
where $\theta^*$ is an unknown column vector, $\beta(x)$ is known and invertible, the $\varphi_i$ are known, $x$ is the state vector of the scalar $x_i$ and the goal is to have $y=x_1$ follow a desired ${n}$ times differentiable $y_r$. The control law propped in \cite{kkkbook} is of the form
\be
u=\frac{1}{\beta(x)}\left(\alpha_n + y_r^{(n)}\right)
\ee
with an update law
\be
\dot\theta = \Gamma W z
\ee
where ${\Gamma=\Gamma^T>0}$ is the adaptive tuning parameter,  $z$ is the transformed state error, and ${W= \tau_{n}(z,\theta)}$ with $\tau_i$ and the ${\alpha_i}$, ${1\leq i\leq n}$ defined in \eqref{app1} in the Appendix along with the rest of the control design. The closed loop system reduces to 
\be
\dot z = A_z(z,\theta,t) z + W(z,\theta,t) \tilde\theta
\ee
where $\tilde\theta=\theta-\theta^*$ and
\ben
A_z=\begin{footnotesize}\bb -c_1 & 1 & 0 &\cdots & 0 \\
-1 & -c_2 & 1+\sigma_{23} & \cdots & \sigma_{2n} \\
0 & -1 -\sigma_{23} &\ddots & \ddots & \vdots \\
\vdots & \vdots & \ddots & \ddots & 1+\sigma_{n-1,n} \\
0 & -\sigma_{2n} & \cdots & -1-\sigma_{n-1,n}& -c_n
\eb\end{footnotesize}\een
where $\sigma_{ik}=-\frac{\partial\alpha_{i-1}}{\partial\theta}\Gamma w_k$. The $c_i$ are free design parameters that arise in the definition of the $\alpha_i$ as defined in \eqref{app1} in the Appendix. Notice that the $-c_i$ act in the same way as the $\ell$ in the simple adaptive system first presented in the reference model in \eqref{eqd:reference}. They act to close the reference trajectories with the plant state. The above system also results in similar L-2 norms for the $z$ error state. Consider the Lyapunov candidate 
\be
V(z(t),\tilde\theta(t))= \frac{1}{2}z^Tz + \frac{1}{2}\tilde\theta^T \Gamma^{-1} \tilde\theta
\ee
which results in a negative semidefinite derivative $\dot V \leq - c_0\norm{z}^2$ where $c_0=\min_{1\leq i\leq n}{(c_i)}$. Thus we can integrate $-\dot V$ to obtain the following bound on the L-2 norm of $z$
\be
\norml{z(t)}2^2 \leq \frac{V(0)}{c_0}.
\ee
It is addressed in \cite[\S 4.4.1]{kkkbook} that while it may appear that increasing $c_0$ uniformly decreases the L-2 norm of $z$, choosing $c_i$ to be a large can result in large $z(0)$. The authors then provide a method for initializing the $z$ dynamics so that $z(0)=0$. We have already discussed why this may not be possible in a real system.

\subsection{Adaptive Control in Robotics}
The control structure presented here is taken directly from \cite[\S 9.2]{slobook}. Consider the dynamics of a rigid manipulator 
\be
H(q) \ddot q + C(q,\dot q) + g(q) = \tau
\ee
where $q$ is the joint angle, and $\tau$ is the torque input. It is assumed that the system can be parameterized as
\be
Y(q,\dot q, \dot q_r, \ddot q_r  ) a=  H(q) \ddot q_r + C(q,\dot q) \dot q_r +  g(q) 
\ee
where $q_r$ is a twice differentiable reference signal, $Y$ is known and $a$ is an unknown vector. The control law is chosen as 
\be
\tau =Y \hat a -k_d s \quad \text{and} \quad
\dot {\hat a} = -\Gamma Y^T s.
\ee
Then, defining the desired dynamics trajectory as $q_d$, the reference dynamics of the system are created by 
\be
\dot q_r = \dot q_d - \lambda \tilde q
\ee
where $\tilde q=q-q_d$ and 
\be
s = \dot q - \dot q_r = \dot{\tilde q} + \lambda \tilde q.
\ee

The stability of the above system can be verified with the following Lyapunov candidate, 
\ben
V= \frac{1}{2}\left( s^T H s + \tilde a^T \Gamma^{-1} \tilde a \right).
\een
Differentiating and using the property that $\dot H=C+C^T$ we have that
\be
\Dot V = - s^T k_d s.
\ee
We note that $\lambda$ has a similar role in this control structure as the $\ell$ in the CRM. The desired trajectory is $q_d$ (like $x_m^o$ in our examples), however the adaptive parameter is updated by the composite variable $s$ instead of directly adjusted by the true reference error. 


We now conjecture as to why closed-loop reference models have not been studied in direct adaptive control until recently. In the two cases of nonlinear adaptive control, closed reference trajectory errors are used to update the adaptive controller. This is performed in the tuning function approach through the selection of the $c_i$ and in the adaptive robot control example with $\lambda$ and the creation of the composite variable $s$. In both cases the stability of the system necessitates the introduction of these variables. In contrast, in model reference adaptive control, stability is derived from the inherent stability of the reference model and hence any addition of new variables becomes superfluous. When no reference model is present, closing the loop on the reference trajectory becomes necessary. With the recent focus on improving transients in adaptive systems, CRM now has a role in MRAC. And as pointed out in this paper, improved transients can result with CRM without introducing peaking by choosing the ratio $\abs\ell /\gamma$ carefully.

\section{Conclusions}
An increasingly oscillatory response with increasing adaptation gain is a transient characteristic that is ubiquitous in all adaptive systems. Recently, a class of adaptive systems has been investigated with closed-loop reference models where such oscillatory response can be minimal. In this paper, a detailed analysis of such adaptive systems is carried out. It is shown through the derivation of analytical bounds on both states of the adaptive system and on parameter derivatives that a phenomenon of peaking can occur with CRMs and that this phenomenon can be curtailed through a combination of design and analysis, with the peaking exponent reduced from 0.5 to zero. In particular, it is shown that bounds on the parameter derivatives can be related to bounds on frequencies and corresponding amplitudes, thereby providing an analytical basis for the transient performance. This guarantees that the resulting adaptive systems have improved transient characteristics with reduced oscillations even as the adaptation gains are increased. CRMs are shown to be implicitly present in other problems including composite control, adaptive nonlinear control and a class of problems in robotics.

\bibliographystyle{IEEEtran}
\bibliography{ref}

\begin{thebibliography}{10}
\providecommand{\url}[1]{#1}
\csname url@samestyle\endcsname
\providecommand{\newblock}{\relax}
\providecommand{\bibinfo}[2]{#2}
\providecommand{\BIBentrySTDinterwordspacing}{\spaceskip=0pt\relax}
\providecommand{\BIBentryALTinterwordstretchfactor}{4}
\providecommand{\BIBentryALTinterwordspacing}{\spaceskip=\fontdimen2\font plus
\BIBentryALTinterwordstretchfactor\fontdimen3\font minus
  \fontdimen4\font\relax}
\providecommand{\BIBforeignlanguage}[2]{{%
\expandafter\ifx\csname l@#1\endcsname\relax
\typeout{** WARNING: IEEEtran.bst: No hyphenation pattern has been}%
\typeout{** loaded for the language `#1'. Using the pattern for}%
\typeout{** the default language instead.}%
\else
\language=\csname l@#1\endcsname
\fi
#2}}
\providecommand{\BIBdecl}{\relax}
\BIBdecl

\bibitem{annbook}
K.~S. Narendra and A.~M. Annaswamy, \emph{Stable Adaptive Systems}.\hskip 1em
  plus 0.5em minus 0.4em\relax Dover, 2005.

\bibitem{ioabook}
P.~Ioannou and J.~Sun, \emph{Robust Adaptive Control}.\hskip 1em plus 0.5em
  minus 0.4em\relax Dover, 2013.

\bibitem{lee97}
T.-G. Lee and U.-Y. Huh, ``An error feedback model based adaptive controller
  for nonlinear systems,'' in \emph{Proceedings of the IEEE International
  Symposium on Industrial Electronics}, 1997.

\bibitem{eug10aiaa}
E.~Lavretsky, ``Adaptive output feedback design using asymptotic properties of
  lqg/ltr controllers,'' in \emph{AIAA 2010--7538}, 2010.

\bibitem{lav12tac}
------, ``Adaptive output feedback design using asymptotic properties of
  lqg/ltr controllers,'' \emph{IEEE Trans. Automat. Contr.}, vol.~57, no.~6,
  2012.

\bibitem{ste10}
V.~Stepanyan and K.~Krishnakumar, ``Mrac revisited: guaranteed perforamance
  with reference model modification,'' in \emph{American Control Conference},
  2010.

\bibitem{ste11}
------, ``M--mrac for nonlinear systems with bounded disturbances,'' in
  \emph{Conference on Decision and Control}, 2011.

\bibitem{gib12}
T.~E. Gibson, A.~M. Annaswamy, and E.~Lavretsky, ``Improved transient response
  in adaptive control using projection algorithms and closed loop reference
  models,'' in \emph{AIAA Guidance Navigation and Control Conference}, 2012.

\bibitem{gib13acc1}
------, ``{Closed--loop Reference Model Adaptive Control, Part I: Transient
  Performance},'' in \emph{American Control Conference}, 2013.

\bibitem{gib13acc2}
------, ``{Closed--loop Reference Model Adaptive Control: Composite control and
  Observer Feedback},'' in \emph{11th IFAC International Workshop on Adaptation
  and Learning in Control and Signal Processing}, 2013.

\bibitem{gib13ecc}
------, ``Closed-loop reference models for output--feedback adaptive systems,''
  in \emph{European Control Conference}, 2013.

\bibitem{krs93}
M.~Krstic and P.~V. Kokotovic, ``Transient--performance improvement with a new
  class of adaptive controllers,'' \emph{Syst. Contr. Lett.}, vol.~21, pp.
  451--461, 1993.

\bibitem{dat94}
A.~Datta and P.~Ioannou, ``Performance analysis and improvement in model
  reference adaptive control,'' \emph{IEEE Trans. Automat. Contr.}, vol.~39,
  no.~12, Dec. 1994.

\bibitem{zan94}
Z.~Zang and R.~Bitmead, ``Transient bounds for adaptive control systems,''
  \emph{IEEE Trans. Automat. Contr.}, vol.~39, no.~1, 1994.

\bibitem{duarte1989combined}
M.~A. Duarte and K.~S. Narendra, ``{Combined direct and indirect approach to
  adaptive control},'' \emph{IEEE Trans. Automat. Contr.}, vol.~34, no.~10, pp.
  1071--1075, 1989.

\bibitem{slotine1989composite}
J.-J. Slotine and W.~Li, ``{Composite adaptive control of robot
  manipulators.}'' \emph{Automatica}, vol.~25, no.~4, pp. 509--519, 1989.

\bibitem{eugeneTAC09}
E.~Lavretsky, ``Combined / composite model reference adaptive control,''
  \emph{IEEE Trans. Automat. Contr.}, vol.~54, no.~11, pp. 2692--2697, 2009.

\bibitem{kkkbook}
M.~Krstic, I.~Kanellakopoulos, and P.~Kokotovic, \emph{Nonlinear and Adaptive
  Control Design}.\hskip 1em plus 0.5em minus 0.4em\relax John Wiley and Sons,
  1995.

\bibitem{slobook}
J.-J. Slotine and W.~Li, \emph{Applied Nonlinear Control}.\hskip 1em plus 0.5em
  minus 0.4em\relax Prentice Hall, 1995.

\bibitem{desbook}
C.~A. Desoar and M.~Vidyasagar, \emph{Feedback systems: input-output
  properties}.\hskip 1em plus 0.5em minus 0.4em\relax Academic Press, 1975.

\bibitem{sus91}
H.~J. Sussmann and P.~V. Kokotovic, ``The peaking phenomonen and the global
  stabilization of nonlinear systems,'' \emph{IEEE Trans. Automat. Contr.},
  vol.~36, no.~4, 1991.

\bibitem{mol03}
C.~Moler and C.~Loan, ``Nineteen dubious ways to compute the matrix exponential
  of a matrix, twenty-five years later,'' \emph{SIAM Review}, 2003.

\bibitem{laxbook}
P.~D. Lax, \emph{Functional Analysis}.\hskip 1em plus 0.5em minus 0.4em\relax
  Wiley-Interscience, 2002.

\bibitem{rud76}
W.~Rudin, \emph{Principles of Mathematical Analysis}.\hskip 1em plus 0.5em
  minus 0.4em\relax McGraw--Hill, 1976.

\bibitem{lig58}
J.~Lighthill, \emph{An Introduction to Fourier Analysis and Generalised
  Functions}.\hskip 1em plus 0.5em minus 0.4em\relax Cambridge University
  Press, 1958.

\bibitem{pom92}
J.~Pomet and L.~Praly, ``Adaptive nonlinear regulation: Estimation from the
  lyapunov equation,'' \emph{IEEE Trans. Automat. Contr.}, vol.~37, no.~6, June
  1992.

\bibitem{bell43}
R.~Bellman, ``The stability of solutions of linear differential equations,''
  \emph{Duke Math J.}, 1943.

\bibitem{mat12}
M.~Matsutani, A.~M. Annaswamy, and E.~Lavretsky, ``Guaranteed delay margins for
  adaptive control of scalar plants,'' in \emph{IEEE, Conference on Decision
  and Control}, 2012.

\bibitem{mat13}
M.~Matsutani, ``Robust adaptive flight control systems in the presence of time
  delay,'' Ph.D. dissertation, Massachusetts Institute of Technology, 2013.

\bibitem{hus13c1}
H.~S. Hussain, M.~Matsutani, A.~M. Annaswamy, and E.~Lavretsky, ``Adaptive
  control of scalar plants in the presence of unmodeled dynamics,'' in
  \emph{IFAC International Workshop on Adaptation and Learning in Control and
  Signal Processing}, 2013.

\bibitem{hus13c2}
------, ``Robust adaptive control in the presence of unmodeled dynamics: A
  counter to rohrs's counterexample,'' in \emph{AIAA Guidance Navigation and
  Control Conference}, 2013.

\end{thebibliography}

\appendices

\section{Tuning Function Parameter Definitions}
\be\label{app1}
\begin{split}
z_i  =& x_i - y_r^{(i-1)}-\alpha_{i-1} \\
\alpha_i =&  -z_{i-1}-c_i z_i - w_i^T \theta \\&+ \sum_{k=1}^{i-1} \left( \frac{\partial \alpha_{i-1}}{\partial x_k}+\frac{\partial \alpha_{i-1}}{\partial y_r^{(k-1)}} y_r^{(k)} \right) \\
&+ \frac{\partial \alpha_{i-1}}{\partial \theta}\Gamma \tau_i + \sum_{k=2}^{i-1}\frac{\partial \alpha_{k-1}}{\partial \theta}\Gamma w_i z_k \\
\tau_i =& \tau_{i-1} + w_i z_i \\
w_i = & \varphi_i - \sum_{k=1}^{i-1}\frac{\partial \alpha_{i-1}}{\partial x_k} \varphi_k
\end{split}
\ee

\end{document}